\documentclass[11pt]{article}

\usepackage[utf8]{inputenc}
\usepackage[T1]{fontenc}
\usepackage{lmodern}
\usepackage[a4paper,margin=2.5cm]{geometry}
\usepackage{microtype}
\usepackage{authblk} 
\usepackage{amsmath,amssymb,amsthm,mathtools}
\usepackage{bm}
\usepackage{stmaryrd}

\usepackage{enumitem}
\setlist[itemize]{leftmargin=2em}
\setlist[enumerate]{leftmargin=2.5em}

\usepackage{tikz}
\usetikzlibrary{calc,positioning,arrows.meta,decorations.pathmorphing}

\usepackage{xcolor}
\usepackage[
colorlinks=true,
linkcolor=blue,
citecolor=blue,
urlcolor=blue
]{hyperref}
\usepackage[nameinlink,noabbrev]{cleveref}
\usepackage{authblk}
\newtheorem{theorem}{Theorem}[section]
\newtheorem{proposition}[theorem]{Proposition}
\newtheorem{lemma}[theorem]{Lemma}
\newtheorem{corollary}[theorem]{Corollary}

\theoremstyle{definition}
\newtheorem{definition}[theorem]{Definition}
\newtheorem{example}[theorem]{Example}

\theoremstyle{remark}
\newtheorem{remark}[theorem]{Remark}

\crefname{section}{Section}{Sections}
\crefname{theorem}{Theorem}{Theorems}
\crefname{proposition}{Proposition}{Propositions}
\crefname{lemma}{Lemma}{Lemmas}
\crefname{corollary}{Corollary}{Corollaries}
\crefname{definition}{Definition}{Definitions}
\crefname{example}{Example}{Examples}
\crefname{remark}{Remark}{Remarks}

\newcommand{\Prop}{\mathsf{Prop}}
\newcommand{\Lang}{\mathcal{L}}
\newcommand{\Fac}{\mathsf{Fac}}

\newcommand{\St}{\mathsf{St}}
\newcommand{\True}[1]{\llbracket #1 \rrbracket}
\newcommand{\Pow}{\mathcal{P}}
\newcommand{\forces}{\models}

\title{A Logic of Secrecy on Simplicial Models}
\author[1]{Shanxia Wang\thanks{Email: wangshanxia@htu.edu.cn}}
\affil[1]{School of Computer Science and Information Engineering(College of Artificial Intelligence), Henan Normal University, Xinxiang, Henan, China}

\begin{document}
\date{} 
\maketitle


\begin{abstract}
	We develop a logic of secrecy on simplicial models for multi-agent systems. Standard simplicial models provide a geometric semantics for knowledge by representing global states as facets of a chromatic simplicial complex and agents' local states as coloured vertices. However, secrecy cannot in general be captured as a genuinely new modality by relying on the ordinary simplicial knowledge structure alone. This motivates the introduction of an additional secrecy layer.
	
	To this end, we define \emph{simplicial secrecy models}, which enrich standard simplicial epistemic models with agent-relative secrecy neighborhood functions attached to local states. On this basis, we introduce a primitive secrecy operator $S_a\varphi$. Semantically, $S_a\varphi$ holds when agent $a$ knows $\varphi$ in the ordinary simplicial sense and, moreover, the truth set of $\varphi$ belongs to one of the designated secrecy neighborhoods associated with $a$'s current local state. The clause for secrecy thus combines an ordinary knowledge requirement with an additional local-state-based neighborhood requirement, while the frame condition ensures that designated secrecy events remain non-trivial from the perspective of every other agent.
	
	We formulate a system $\mathsf{SSL}$ for the resulting language and show that it is sound with respect to the class of simplicial secrecy models. For the genuinely multi-agent case $|A|\ge 2$, we prove completeness by first constructing an auxiliary-colour canonical model and then representing it inside the original class of pure $A$-chromatic simplicial secrecy models. The resulting framework yields a primitive, local-state-based, and geometrically grounded account of secrecy on simplicial models, together with a sound axiomatization and, in the genuinely multi-agent case, a complete one.
\end{abstract}


\section{Introduction}\label{sec:introduction}

\subsection{Motivation: secrecy beyond ordinary epistemic knowledge}

Secrecy is a central notion in epistemic reasoning about multi-agent systems.
In distributed computation, security protocols, and information flow, one often
needs to express that some piece of information is available to one agent while
remaining hidden from others~\cite{halpern1990common,fagin1995reasoning,
	fagin1992machines}. In standard epistemic logic, such situations are usually
analyzed in possible-worlds frameworks, most notably in Kripke semantics, where
secrecy is represented by combining one agent's knowledge with the ignorance of
other agents~\cite{halpern2002secrecy,halpern2005anonymity,dechesne2010security}.
A closely related notion is that of \emph{exclusive knowledge}, studied in the
context of secrecy logics by Xiong and \r{A}gotnes~\cite{xiong2023secrets,
	xiong2026solitary}, where a proposition counts as secret for agent $a$ if $a$
knows it and no other agent does. The epistemic and logical analysis of secrecy
has also been explored from the perspective of intentional secret-keeping~\cite{aldini2024logical},
and in connection with concrete epistemic puzzles such as the Russian Cards
Problem~\cite{vanditmarsch2003russiancards,rajsbaum2023russian,leyva2021russian}.

These approaches are natural and mathematically fruitful, but they leave open a
structural question that becomes especially significant in multi-agent settings:
what does secrecy look like when epistemic states are not treated as atomic
worlds, but as combinations of local states? In ordinary Kripke-style secrecy
logics, secrecy is typically not a genuinely new semantic notion. Instead, it is
usually definable as a combination of one agent's knowledge and the failure of
other agents' knowledge, for example in the form
\[
K_a\varphi \land \bigwedge_{b\neq a}\neg K_b\varphi.
\]
If epistemic states are instead assembled from local pieces of information, then
one may ask whether secrecy should still be understood merely as such a
definable combination, or whether it admits a more structural and genuinely
local interpretation.

\subsection{Why simplicial semantics}

Simplicial semantics provides exactly the kind of framework in which this
question becomes meaningful. The use of simplicial complexes in distributed
computing goes back to the foundational work of Herlihy and
Shavit~\cite{herlihy1999topological}, who established the topological structure
of asynchronous computability, a line of research developed into a comprehensive
combinatorial-topological approach to distributed systems
in~\cite{herlihy2014distributed}. Building on this foundation, Goubault, Ledent,
and Rajsbaum introduced simplicial epistemic models as a geometric semantics
for epistemic logic~\cite{goubault2018simplicial,goubault2021simplicial}, with
a systematic development in the doctoral thesis of Ledent~\cite{ledent2019geometric}
and a comprehensive survey given by van Ditmarsch, Goubault, Ledent, and
Rajsbaum~\cite{vanditmarsch2022knowledge}.

In simplicial epistemic models, global states are represented by facets of a
chromatic simplicial complex, while vertices encode agents' local states.
Epistemic accessibility is induced by local-state sharing: two facets are
indistinguishable for agent $a$ whenever they contain the same $a$-coloured
vertex. Thus knowledge is not imposed by an external accessibility relation, but
emerges from the local-state structure itself. This local-state perspective has
proved fruitful for the interpretation of knowledge and its
dynamics~\cite{goubault2018simplicial,goubault2021simplicial,ledent2019geometric,
	vanditmarsch2021equality,goubault2022kb4n}, as well as for the study of
distributed knowledge~\cite{goubault2023semisimplicial,galimullin2024varieties,
	balbiani2024dynamicdistributed}, and for extensions to more general simplicial
frameworks such as impure complexes and related
structures~\cite{randrianomentsoa2023impure,goubault2024hypergraphs}.

It is therefore natural to ask whether secrecy can also be given a genuinely
simplicial treatment. However, this requires some care. If one tries to define
secrecy only in terms of the ordinary simplicial knowledge structure, then
secrecy does not become a genuinely new modality. It simply reduces to the same
knowledge-and-ignorance pattern familiar from Kripke-based secrecy logics and
acquires no independent semantic status of its own. In other words, if secrecy
is to be studied as a primitive modal notion on simplicial models, it must be
supported by an additional semantic layer beyond the standard simplicial
epistemic structure~\cite{goubault2021simplicial,vanditmarsch2022knowledge}.

\subsection{Main idea of the paper}

The main idea of this paper is to preserve the ordinary simplicial semantics for
knowledge while adding an extra, local-state-based semantic layer for secrecy.
More precisely, we enrich standard chromatic simplicial epistemic models with
agent-relative secrecy neighborhood functions attached to vertices, that is, to
agents' local states. In this way, secrecy is anchored directly in the local
geometry of the simplicial complex rather than treated as a purely abstract
non-normal modality.

On this basis, we introduce a primitive secrecy operator $S_a\varphi$.
Semantically,
\[
S_a\varphi
\]
is interpreted as the conjunction of two requirements: first, agent $a$ knows
$\varphi$ in the ordinary simplicial sense, and second, the truth set of
$\varphi$ belongs to one of the secrecy neighborhoods associated with $a$'s
current local state. In slogan form, the intended reading is:
\[
S_a\varphi
\quad\text{means}\quad
K_a\varphi \text{ plus a local-state-based secrecy designation.}
\]
The knowledge conjunct guarantees that all of $a$'s epistemic possibilities
satisfy $\varphi$, while the neighborhood component selects which such known
events count as secrets.

A further frame condition, denoted by \textbf{(SN)}, ensures that any designated
secret event remains epistemically non-trivial from the perspective of every
other agent. Intuitively, if an event is designated as secret for agent $a$ at a
given local state, then every other agent must still consider some
indistinguishable alternative outside that event. As a consequence, whenever
$\varphi$ is secret for $a$, agent $a$ knows $\varphi$, $\varphi$ is true, and
every other agent fails to know both $\varphi$ and $\neg\varphi$. This
epistemic profile parallels familiar conclusions from Kripke-based secrecy
logics~\cite{xiong2023secrets,halpern2002secrecy,halpern2005anonymity,
	aldini2024logical}, but here it emerges from a genuinely simplicial and
vertex-based semantics.

Conceptually, our framework occupies a middle ground between two familiar
perspectives. On the one hand, it preserves the standard simplicial
interpretation of epistemic knowledge~\cite{goubault2021simplicial,
	vanditmarsch2022knowledge} and therefore retains the local-state and
geometric intuitions characteristic of simplicial semantics. On the other hand,
it treats secrecy as a primitive operator whose semantics depends on an
additional neighborhood layer, in the spirit of neighborhood and evidence-based
epistemic logics~\cite{chellas1980modal,vanbenthem2014evidence,
	vanbenthem2012evidence,wang2024psnl}. In this sense, our framework is not
merely a reformulation of knowledge-based secrecy~\cite{xiong2023secrets,
	xiong2026solitary}, but a structured simplicial semantics for secrecy in its
own right.

\subsection{Main results and contributions}

The contribution of the paper is therefore not simply to transplant an existing
secrecy notion into a new setting, but to develop a logic of secrecy that is
simultaneously epistemic, simplicial, and primitive. The main contributions are
as follows.
\begin{enumerate}
	\item We introduce \emph{simplicial secrecy models}, extending standard
	$A$-chromatic simplicial epistemic models with agent-relative secrecy
	neighborhood functions attached to local states.
	
	\item We define a primitive secrecy operator $S_a$ whose semantics is
	extensional and vertex-based: $S_a\varphi$ requires that agent $a$ knows
	$\varphi$ and that the truth set of $\varphi$ is designated as a secrecy
	event at $a$'s current local state.
	
	\item We formulate the system $\mathsf{SSL}$ for the resulting language,
	incorporating the owner-locality principle forced by the semantics, and
	derive additional valid principles characterizing secrecy.
	
	\item For the genuinely multi-agent case $|A|\ge 2$, we prove soundness and
	completeness of $\mathsf{SSL}$ with respect to simplicial secrecy models.
	The completeness proof proceeds by first constructing an auxiliary-colour
	canonical model and then representing it inside the original class of pure
	$A$-chromatic simplicial secrecy models.
\end{enumerate}

The resulting framework provides a new foundation for the study of secrecy in
simplicial epistemic logic. Conceptually, it shows how secrecy can be added as
a genuine new semantic layer on top of standard simplicial
knowledge~\cite{goubault2021simplicial,vanditmarsch2022knowledge,
	ledent2019geometric}. Technically, for the genuinely multi-agent case
$|A|\ge 2$, it yields a sound and complete axiomatization of secrecy on
simplicial models. The technical approach combines the canonical model method
for neighborhood logics~\cite{chellas1980modal,vanbenthem2014evidence,
	vanbenthem2012evidence} with the simplicial representation technique
characteristic of this line of research~\cite{goubault2021simplicial,
	goubault2022kb4n,goubault2023semisimplicial,ledent2019geometric}.

The remainder of the paper is organized as follows. Section~\ref{sec:relatedwork}
discusses related work. Section~\ref{sec:syntax-semantics} introduces the
syntactic and semantic framework, including simplicial secrecy models and their
truth conditions. Section~\ref{sec:axiomatic-system} presents the axiomatic
system $\mathsf{SSL}$ and its basic proof-theoretic properties.
Section~\ref{sec:metatheory} establishes soundness and completeness, for
$|A|\ge 2$, by combining an auxiliary-colour canonical construction with a
representation theorem back into pure $A$-chromatic simplicial secrecy models.
Section~\ref{sec:conclusion} concludes.


\section{Related Work}\label{sec:relatedwork}

The present paper sits at the intersection of three research streams: the
logical analysis of secrecy in multi-agent systems, simplicial models for
epistemic logic, and neighborhood semantics for non-normal modalities. This
section is intended only to locate the present contribution within these
literatures and to indicate the main point of departure from each of them.

\subsection{Logical analyses of secrecy}

The epistemic analysis of secrecy in multi-agent systems has a substantial
history. The foundational framework for knowledge in distributed systems was
laid by Halpern and Moses~\cite{halpern1990common} and developed
comprehensively in the textbook by Fagin, Halpern, Moses, and
Vardi~\cite{fagin1995reasoning}; the properties of knowledge specific to
distributed environments were further studied in~\cite{fagin1992machines}.
Within this tradition, Halpern and O'Neill introduced formal treatments of
secrecy and anonymity in multiagent systems~\cite{halpern2002secrecy,
	halpern2005anonymity}, defining secrecy in terms of what one agent knows and
what other agents fail to know. Their approach is Kripke-based: secrecy is a
condition expressible by combining knowledge and ignorance, rather than a
separate semantic primitive.

The most directly related logical work is that of Xiong and
\r{A}gotnes~\cite{xiong2023secrets,xiong2026solitary}. In~\cite{xiong2023secrets},
they develop a logic of secrets in which a proposition $\varphi$ is secret for
agent $a$ if $a$ knows $\varphi$ and no other agent does, that is,
\[
S_a\varphi := K_a\varphi \land \bigwedge_{b\neq a}\neg K_b\varphi.
\]
In that framework, secrecy is again a \emph{definable} combination of
knowledge and ignorance, and the main technical results concern the
interpolation properties of the resulting logic. The follow-up
work~\cite{xiong2026solitary} further investigates the relationship between
exclusive knowledge and secrecy logics.

A different line is pursued by Aldini et al.~\cite{aldini2024logical}, who
study the intentional dimension of secret-keeping by combining knowledge,
belief, and intention. Their interest lies in the mental states of the
secret-keeper rather than in the local-state geometry of a distributed
epistemic model.

The Russian Cards Problem and related secrecy puzzles provide natural
benchmarks for epistemic analyses of hidden information~\cite{vanditmarsch2003russiancards}.
From a distributed computing perspective, Rajsbaum connects unconditionally
secure information transmission in Russian cards problems to simplicial
models~\cite{rajsbaum2023russian}, while Leyva-Acosta, Pascual-Aseff, and
Rajsbaum study protocol aspects of information exchange in that
setting~\cite{leyva2021russian}. These works illustrate the importance of
secrecy in settings where agents hold partial local information, but they do
not provide a primitive simplicial semantics for secrecy itself.

Taken together, these approaches show that secrecy has been extensively studied
in epistemic logic, but most existing analyses treat it not as a primitive
modality, but as a pattern built from knowledge and ignorance. The present
paper departs from that tradition by giving secrecy an independent, vertex-based
semantics on simplicial models.

\subsection{Simplicial models for epistemic logic}

The use of simplicial complexes in distributed computing originates with the
topological characterization of asynchronous computability by Herlihy and
Shavit~\cite{herlihy1999topological}, with a comprehensive treatment in the
monograph of Herlihy, Kozlov, and Rajsbaum~\cite{herlihy2014distributed}. In
this setting, protocol complexes represent collections of compatible local
states, and facets represent complete global configurations.

The connection between simplicial complexes and epistemic logic was established
by Goubault, Ledent, and Rajsbaum in a series of foundational
papers~\cite{goubault2018simplicial,goubault2021simplicial}, with a systematic
development in Ledent's doctoral thesis~\cite{ledent2019geometric} and a
comprehensive survey in van Ditmarsch, Goubault, Ledent, and
Rajsbaum~\cite{vanditmarsch2022knowledge}. In this framework, chromatic
simplicial complexes provide a geometric semantics for S5 multi-agent
epistemic logic: vertices represent agents' local states, facets represent
global states, and epistemic indistinguishability is determined by vertex
identity. Thus the epistemic relation $\sim_a$, and hence ordinary knowledge,
emerges from local-state sharing rather than from an independently stipulated
accessibility relation.

This line of research has been extended in several directions. Dynamic
epistemic analyses of equality negation and related covering tasks appear
in~\cite{vanditmarsch2021equality}. The logic $KB4_n$, modeling agents that may
crash and lose their local state, was given a simplicial semantics
in~\cite{goubault2022kb4n}. Distributed and semi-simplicial set models for
distributed knowledge were developed in~\cite{goubault2023semisimplicial},
with further variations studied in~\cite{galimullin2024varieties,
	balbiani2024dynamicdistributed}. The simplicial perspective has also been
extended beyond the standard pure setting, for example to impure simplicial
complexes and to chromatic hypergraphs~\cite{randrianomentsoa2023impure,
	goubault2024hypergraphs}.

What unifies this literature is that it studies knowledge and its variants in
geometric, local-state-based terms. By contrast, the present paper adds a
genuinely new modality on top of the standard simplicial epistemic base.
Rather than varying the treatment of knowledge itself, we introduce secrecy as
a primitive operator supported by an additional semantic layer.

\subsection{Neighborhood semantics for non-normal modalities}

Neighborhood semantics, introduced by Chellas~\cite{chellas1980modal}, provides
a general framework for non-normal modalities. Instead of interpreting a modal
operator by means of an accessibility relation, one assigns to each state a
collection of sets of states, its neighborhood, and a modal formula is true at
that state when the truth set of the formula belongs to the neighborhood.

In the epistemic setting, neighborhood structures have been used to model
evidence and belief~\cite{vanbenthem2012evidence,vanbenthem2014evidence},
topological reasoning about knowledge~\cite{moss1992topological}, subset-space
logic~\cite{wang2013subset}, and more recent frameworks such as point-set
neighborhood logic~\cite{wang2024psnl}. In distributed and
communication-oriented settings, related ideas also appear in dynamic models of
communication patterns and epistemic information flow~\cite{velazquez2021communication,
	castaneda2024pattern}.

Our use of neighborhoods is continuous with this tradition, but differs from it
in a crucial respect. In standard neighborhood semantics, neighborhoods are
typically attached to worlds or global states. In the present paper, by
contrast, secrecy neighborhoods are attached to \emph{vertices}, that is, to
agents' local states in a simplicial complex. This makes the additional modal
layer genuinely local-state-based. Moreover, owner-knowledge is imposed
directly in the truth clause for $S_a$, while the frame condition
\textbf{(SN)} is tailored specifically to the secrecy reading: it ensures that
any designated secret event remains epistemically non-trivial from the
perspective of every other agent at every facet compatible with the owner's
local state.

Thus neighborhood semantics provides the right semantic technology for a
primitive secrecy operator, but the present framework adapts that technology to
the geometry of simplicial models by anchoring the neighborhood layer in local
states rather than in abstract worlds.

\subsection{Position of the present work}

The present paper occupies a distinctive position relative to these three
traditions. Relative to Kripke-based secrecy logics~\cite{halpern2002secrecy,
	halpern2005anonymity,xiong2023secrets,xiong2026solitary}, our framework
treats secrecy as a \emph{primitive} modality rather than as a definable
combination of knowledge and ignorance. Relative to existing simplicial
epistemic work~\cite{goubault2021simplicial,vanditmarsch2022knowledge,
	goubault2022kb4n,goubault2023semisimplicial}, we add a genuinely new modal
layer instead of extending or varying the treatment of knowledge itself.
Finally, relative to standard neighborhood semantics~\cite{chellas1980modal,
	vanbenthem2014evidence,moss1992topological,wang2024psnl}, our neighborhoods
are not world-based but vertex-based, and are therefore anchored directly in
the local-state geometry of simplicial complexes.

In this sense, the contribution of the paper is to bring together three ideas
that had previously remained separate: secrecy as an epistemic phenomenon,
simplicial models as a semantics of local states, and neighborhood structure as
a semantics for non-normal modalities. The result is a primitive,
vertex-based, and geometrically grounded logic of secrecy on simplicial
models.


\section{Syntax and Semantics}\label{sec:syntax-semantics}

In this section we introduce the formal framework of the paper. We begin with the
standard simplicial semantics for epistemic knowledge, where global states are
represented by facets of a chromatic simplicial complex and an agent's knowledge
is determined by the local state represented by her coloured vertex. We then add
a secrecy layer, consisting of agent-relative secrecy neighborhoods attached to
local states. This yields a primitive semantics for secrecy on top of the
ordinary simplicial epistemic base.

\subsection{Underlying simplicial epistemic structures}

Let $A$ be a finite non-empty set of agents, and let $\Prop$ be a countable set
of propositional variables.

\begin{definition}[Simplicial complex]
	A \emph{simplicial complex} is a pair $\mathcal{C}=(V,\mathcal{F})$ such that:
	\begin{enumerate}
		\item $V$ is a non-empty set of vertices;
		\item $\mathcal{F}\subseteq \mathcal{P}_{\mathrm{fin}}(V)$ is a non-empty family of finite non-empty subsets of $V$, called \emph{faces};
		\item $\mathcal{F}$ is downward closed: if $X\in\mathcal{F}$ and $\emptyset\neq Y\subseteq X$, then $Y\in\mathcal{F}$.
	\end{enumerate}
\end{definition}

Intuitively, a face represents a compatible partial configuration of local
states. In epistemic applications, maximal faces represent full global states, a
perspective grounded in the topological approach to distributed
computing~\cite{herlihy1999topological,herlihy2014distributed}.

\begin{definition}[Facets]
	Let $\mathcal{C}=(V,\mathcal{F})$ be a simplicial complex. A face
	$X\in\mathcal{F}$ is a \emph{facet} if it is maximal under inclusion, that is,
	if whenever $X\subseteq Y\in\mathcal{F}$, then $X=Y$. We write
	$\Fac(\mathcal{C})$ for the set of all facets of $\mathcal{C}$.
\end{definition}

\begin{definition}[Chromatic simplicial complex]
	An \emph{$A$-chromatic simplicial complex} is a triple
	\[
	\mathcal{C}=(V,\mathcal{F},\chi)
	\]
	such that:
	\begin{enumerate}
		\item $(V,\mathcal{F})$ is a simplicial complex;
		\item $\chi:V\to A$ is a colouring function;
		\item for every face $X\in\mathcal{F}$, the restriction $\chi|_X$ is injective.
	\end{enumerate}
\end{definition}

The injectivity condition ensures that no face contains two distinct local states
of the same agent, formalizing the principle that each agent occupies a unique
local state in any configuration; compare the standard chromatic setting in
simplicial semantics for distributed systems and epistemic
logic~\cite{herlihy2014distributed,goubault2018simplicial,goubault2021simplicial}.

\begin{definition}[Chromatic simplicial epistemic model]
	An \emph{$A$-chromatic simplicial epistemic model} is a tuple
	\[
	M=(V,\mathcal{F},\chi,\nu)
	\]
	such that:
	\begin{enumerate}
		\item $(V,\mathcal{F},\chi)$ is an $A$-chromatic simplicial complex;
		\item every facet $X\in\Fac(M)$ satisfies $\chi[X]=A$;
		\item every vertex belongs to some facet of $M$, i.e.
		\[
		V=\bigcup_{X\in\Fac(M)} X;
		\]
		\item $\nu:\Fac(M)\to\mathcal{P}(\Prop)$ is a valuation function.
	\end{enumerate}
\end{definition}

Thus every facet contains exactly one vertex of each colour, and every vertex
occurs in at least one facet. Hence every local state represented in the model
is globally realizable in some complete configuration, as in the standard pure
simplicial epistemic semantics~\cite{goubault2018simplicial,
	goubault2021simplicial,ledent2019geometric,vanditmarsch2022knowledge}.

Since every facet contains exactly one vertex of each colour, for every agent
$a\in A$ and every facet $X\in\Fac(M)$ there is a unique vertex of colour $a$
contained in $X$. We denote this vertex by
\[
v_a(X).
\]

\begin{definition}[Star]
	Let $M=(V,\mathcal{F},\chi,\nu)$ be an $A$-chromatic simplicial epistemic
	model. For every vertex $v\in V$, define
	\[
	\St(v)=\{X\in\Fac(M)\mid v\in X\}.
	\]
\end{definition}

So $\St(v)$ is the set of all facets containing the local state represented by
$v$. In the simplicial epistemic literature, stars of vertices play exactly the
role of local epistemic ranges determined by an agent's local
state~\cite{goubault2021simplicial,ledent2019geometric,
	vanditmarsch2022knowledge}.

\begin{definition}[Epistemic indistinguishability]\label{def:sim-a}
	Let $M=(V,\mathcal{F},\chi,\nu)$ be an $A$-chromatic simplicial epistemic
	model, let $a\in A$, and let $X,Y\in\Fac(M)$. We define
	\[
	X\sim_a Y
	\quad\text{iff}\quad
	v_a(X)=v_a(Y).
	\]
	Equivalently,
	\[
	X\sim_a Y
	\quad\text{iff}\quad
	Y\in\St(v_a(X)).
	\]
\end{definition}

This is the standard simplicial notion of epistemic accessibility~\cite{goubault2018simplicial,
	goubault2021simplicial,ledent2019geometric,vanditmarsch2022knowledge}: two
global states are indistinguishable for agent $a$ exactly when they contain the
same local state of $a$. This local-state-based characterization of epistemic
indistinguishability is central to the simplicial approach to knowledge in
distributed systems~\cite{herlihy2014distributed}.

\begin{lemma}\label{lem:sim-equiv}
	For every agent $a\in A$, the relation $\sim_a$ on $\Fac(M)$ is an
	equivalence relation.
\end{lemma}

\begin{proof}
	Reflexivity is immediate, since every facet $X$ contains the vertex $v_a(X)$
	and hence $v_a(X)=v_a(X)$. Symmetry is obvious. For transitivity, suppose
	$X\sim_a Y$ and $Y\sim_a Z$. Then $v_a(X)=v_a(Y)$ and $v_a(Y)=v_a(Z)$,
	hence $v_a(X)=v_a(Z)$, so $X\sim_a Z$.
\end{proof}

\subsection{Simplicial secrecy models}

We retain the standard simplicial semantics for knowledge, but enrich the model
with an additional secrecy layer. Intuitively, this extra structure specifies,
for each agent-local state, which epistemic events count as \emph{secret events}
for that agent at that local state. The use of neighborhood functions to capture
such additional semantic structure follows the established tradition in modal
logic for modeling non-normal modalities~\cite{chellas1980modal} and, more
specifically, their epistemic use in evidence-based
frameworks~\cite{vanbenthem2012evidence,vanbenthem2014evidence}.

\begin{definition}[Simplicial secrecy model]\label{def:ssm}
	A \emph{simplicial secrecy model} is a tuple
	\[
	M=(V,\mathcal{F},\chi,\nu,\{N_a^S\}_{a\in A})
	\]
	such that:
	\begin{enumerate}
		\item $(V,\mathcal{F},\chi,\nu)$ is an $A$-chromatic simplicial epistemic model;
		\item for each agent $a\in A$, let
		\[
		V_a=\{v\in V \mid \chi(v)=a\}.
		\]
		Then
		\[
		N_a^S : V_a \to \mathcal{P}\bigl(\mathcal{P}(\Fac(M))\bigr)
		\]
		is a \emph{secrecy neighborhood function};
		\item for every agent $a\in A$, every $v\in V_a$, and every
		$U\in N_a^S(v)$, the following condition holds:
		\begin{enumerate}
			\item[\textbf{(SN)}] \emph{External uncertainty:} for every facet
			$X\in \St(v)$ and every agent $b\in A\setminus\{a\}$, there exists
			a facet $Y\in \Fac(M)$ such that
			\[
			X\sim_b Y
			\qquad\text{and}\qquad
			Y\notin U.
			\]
		\end{enumerate}
	\end{enumerate}
\end{definition}

\begin{remark}
	The knowledge requirement for secrecy is imposed in the truth clause for
	$S_a$, not in the frame conditions. Thus secrecy is modeled explicitly as the
	conjunction of ordinary simplicial knowledge and a local-state-based
	neighborhood designation. This avoids building owner-knowledge twice into the
	semantics.
\end{remark}

\subsection{Language and truth conditions}

We work in the language
\[
\varphi ::= p \mid \neg \varphi \mid (\varphi \land \varphi) \mid K_a\varphi \mid S_a\varphi,
\qquad p\in\Prop,\ a\in A.
\]
We write this language as $\Lang_{KS}$. As usual, $\lor$, $\rightarrow$,
$\leftrightarrow$, $\top$, and $\bot$ are defined by abbreviation.

Formulas are evaluated at facets.

\begin{definition}[Truth conditions]\label{def:ssm-semantics}
	Let
	\[
	M=(V,\mathcal{F},\chi,\nu,\{N_a^S\}_{a\in A})
	\]
	be a simplicial secrecy model, and let $X\in \Fac(M)$. The satisfaction
	relation $M,X\forces \varphi$ is defined inductively as follows:
	\begin{align*}
		M,X &\forces p
		&&\text{iff } p\in \nu(X),\\
		M,X &\forces \neg \varphi
		&&\text{iff } M,X \not\forces \varphi,\\
		M,X &\forces \varphi \land \psi
		&&\text{iff } M,X\forces \varphi \text{ and } M,X\forces \psi,\\
		M,X &\forces K_a\varphi
		&&\text{iff } \forall Y\in\Fac(M)\, (X\sim_a Y \Rightarrow M,Y\forces \varphi),\\
		M,X &\forces S_a\varphi
		&&\text{iff } M,X\forces K_a\varphi
		\text{ and } \True{\varphi}\in N_a^S(v_a(X)),
	\end{align*}
	where
	\[
	\True{\varphi}=\{Y\in \Fac(M)\mid M,Y\forces \varphi\}.
	\]
\end{definition}

\begin{remark}
	The clause for $S_a$ is extensional in the sense that it depends on the truth
	set $\True{\varphi}$ rather than on the syntactic form of $\varphi$.
	Conceptually, $S_a\varphi$ expresses that agent $a$ knows $\varphi$ and that
	the event $\True{\varphi}$ is designated as secret at $a$'s current local
	state.
\end{remark}

\begin{lemma}[Owner-local normalization]\label{lem:owner-local-normalization}
	Let
	\[
	M=(V,\mathcal{F},\chi,\nu,\{N_a^S\}_{a\in A})
	\]
	be a simplicial secrecy model. For each agent $a\in A$ and each $v\in V_a$,
	define
	\[
	N_a^{S,\mathrm{loc}}(v)
	:=
	\{U\in N_a^S(v)\mid \St(v)\subseteq U\}.
	\]
	Let
	\[
	M^{\mathrm{loc}}
	=
	(V,\mathcal{F},\chi,\nu,\{N_a^{S,\mathrm{loc}}\}_{a\in A}).
	\]
	Then $M^{\mathrm{loc}}$ is again a simplicial secrecy model. Moreover, for
	every formula $\varphi\in\Lang_{KS}$ and every facet $X\in\Fac(M)$,
	\[
	M,X\forces \varphi
	\quad\Longleftrightarrow\quad
	M^{\mathrm{loc}},X\forces \varphi.
	\]
\end{lemma}

\begin{proof}
	First, $M^{\mathrm{loc}}$ is again a simplicial secrecy model. Indeed, for
	each agent $a\in A$ and each $v\in V_a$, we have
	\[
	N_a^{S,\mathrm{loc}}(v)\subseteq N_a^S(v).
	\]
	Hence every set in $N_a^{S,\mathrm{loc}}(v)$ already satisfies \textbf{(SN)}
	in $M$, and therefore also in $M^{\mathrm{loc}}$.
	
	We now prove by induction on the structure of $\varphi$ that for every
	$X\in\Fac(M)$,
	\[
	M,X\forces \varphi
	\quad\Longleftrightarrow\quad
	M^{\mathrm{loc}},X\forces \varphi.
	\]
	
	The propositional cases are immediate, since the valuation is unchanged. The
	case of $K_a\psi$ is also immediate, because the underlying simplicial
	epistemic structure is unchanged.
	
	It remains to consider the case $\varphi=S_a\psi$. Let $X\in\Fac(M)$ and
	write
	\[
	v:=v_a(X).
	\]
	By the induction hypothesis, for every facet $Y\in\Fac(M)$,
	\[
	M,Y\forces \psi
	\quad\Longleftrightarrow\quad
	M^{\mathrm{loc}},Y\forces \psi.
	\]
	Hence
	\[
	\llbracket\psi\rrbracket_M
	=
	\llbracket\psi\rrbracket_{M^{\mathrm{loc}}}.
	\]
	Since the underlying simplicial epistemic structure is the same in both
	models, we also have
	\[
	M,X\forces K_a\psi
	\quad\Longleftrightarrow\quad
	M^{\mathrm{loc}},X\forces K_a\psi.
	\]
	
	Suppose first that
	\[
	M,X\forces S_a\psi.
	\]
	Then
	\[
	M,X\forces K_a\psi
	\qquad\text{and}\qquad
	\llbracket\psi\rrbracket_M\in N_a^S(v).
	\]
	Because $M,X\forces K_a\psi$, every facet in $\St(v)$ satisfies $\psi$, so
	\[
	\St(v)\subseteq \llbracket\psi\rrbracket_M.
	\]
	Therefore
	\[
	\llbracket\psi\rrbracket_M\in N_a^{S,\mathrm{loc}}(v).
	\]
	Together with
	\[
	M^{\mathrm{loc}},X\forces K_a\psi,
	\]
	this yields
	\[
	M^{\mathrm{loc}},X\forces S_a\psi.
	\]
	
	Conversely, suppose that
	\[
	M^{\mathrm{loc}},X\forces S_a\psi.
	\]
	Then
	\[
	M^{\mathrm{loc}},X\forces K_a\psi
	\qquad\text{and}\qquad
	\llbracket\psi\rrbracket_{M^{\mathrm{loc}}}\in N_a^{S,\mathrm{loc}}(v).
	\]
	By definition of $N_a^{S,\mathrm{loc}}(v)$, this implies
	\[
	\llbracket\psi\rrbracket_{M^{\mathrm{loc}}}\in N_a^S(v).
	\]
	Using
	\[
	\llbracket\psi\rrbracket_M
	=
	\llbracket\psi\rrbracket_{M^{\mathrm{loc}}}
	\]
	and
	\[
	M,X\forces K_a\psi
	\quad\Longleftrightarrow\quad
	M^{\mathrm{loc}},X\forces K_a\psi,
	\]
	we obtain
	\[
	M,X\forces S_a\psi.
	\]
	
	This completes the induction.
\end{proof}

\subsection{A running example}\label{subsec:running-example}

We will use the following two-agent example throughout the paper. It already
illustrates the simplicial notions of local star and epistemic
indistinguishability, the secrecy frame condition \textbf{(SN)}, and the truth
of a basic secrecy formula. Later, by changing only the valuation and the
secrecy neighborhoods while keeping the same underlying simplicial complex, we
will reuse the same geometry for the countermodels showing that secrecy is
non-normal.

\begin{figure}[t]
	\centering
	\begin{tikzpicture}[x=1.8cm,y=1.2cm,every node/.style={font=\small}]
		\fill[blue!15] (0.5,2) rectangle (1.5,3);
		\fill[blue!15] (1.5,2) rectangle (2.5,3);
		\fill[blue!15] (2.5,2) rectangle (3.5,3);
		\fill[blue!15] (0.5,1) rectangle (1.5,2);
		
		\draw[thick] (0.5,0) rectangle (3.5,3);
		\draw[thick] (1.5,0) -- (1.5,3);
		\draw[thick] (2.5,0) -- (2.5,3);
		\draw[thick] (0.5,1) -- (3.5,1);
		\draw[thick] (0.5,2) -- (3.5,2);
		
		\node[left] at (0.5,2.5) {$u_0$};
		\node[left] at (0.5,1.5) {$u_1$};
		\node[left] at (0.5,0.5) {$u_2$};
		
		\node[above] at (1,3) {$w_1$};
		\node[above] at (2,3) {$w_2$};
		\node[above] at (3,3) {$w_3$};
		
		\node at (1,2.5) {$x_1$};
		\node at (2,2.5) {$x_2$};
		\node at (3,2.5) {$x_3$};
		\node at (1,1.5) {$y_1$};
		\node at (2,1.5) {$y_2$};
		\node at (3,1.5) {$y_3$};
		\node at (1,0.5) {$z_1$};
		\node at (2,0.5) {$z_2$};
		\node at (3,0.5) {$z_3$};
		
		\draw[blue,very thick,rounded corners] (0.45,1.95) rectangle (3.55,3.05);
		
		\draw[red,very thick] (1,0.5) circle (0.14);
		\draw[red,very thick] (2,1.5) circle (0.14);
		\draw[red,very thick] (3,1.5) circle (0.14);
		
		\node[align=left,right] at (5.6,2.3) {\textcolor{blue}{Blue box:} $\St(u_0)$};
		\node[align=left,right] at (5.6,1.7) {\textcolor{red}{Red circles:} witnesses};
		\node[align=left,right] at (5.6,1.1) {Shaded cells: $U$};
	\end{tikzpicture}
	\caption{A two-agent running example. Rows are $a$-equivalence classes and
		columns are $b$-equivalence classes. The shaded event is
		\(
		U=\{x_1,x_2,x_3,y_1\}.
		\)
		The top row is \(\St(u_0)\). In each column there is a facet outside
		\(U\), namely \(z_1\), \(y_2\), and \(y_3\), witnessing condition
		\textbf{(SN)} for the owner \(a\) at the local state \(u_0\).}
	\label{fig:running-example}
\end{figure}
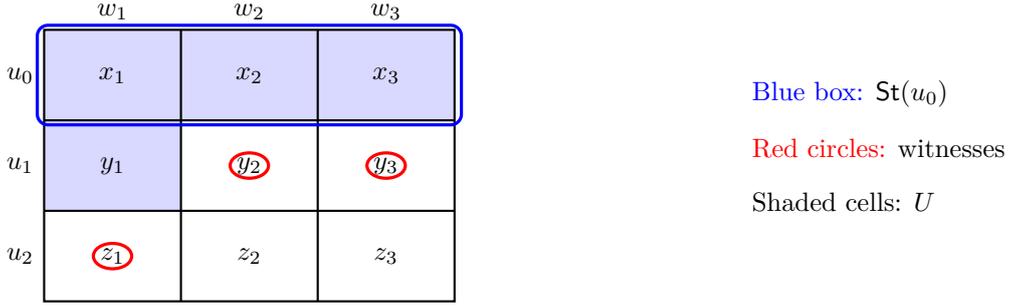

\begin{example}[A two-agent secrecy model]\label{ex:running-example}
	Let
	\[
	A=\{a,b\}.
	\]
	Consider the pure $A$-chromatic simplicial complex with vertices
	\[
	u_0,u_1,u_2 \text{ of colour } a,
	\qquad
	w_1,w_2,w_3 \text{ of colour } b,
	\]
	and facets
	\[
	x_i=\{u_0,w_i\},\qquad
	y_i=\{u_1,w_i\},\qquad
	z_i=\{u_2,w_i\}
	\qquad (i=1,2,3).
	\]
	
	Let
	\[
	U=\{x_1,x_2,x_3,y_1\}.
	\]
	Define the secrecy neighborhoods by
	\[
	N_a^S(u_0)=\{U\},
	\]
	and let all other secrecy neighborhoods be empty.
	
	Let the propositional variable $p$ be true exactly on $U$, i.e.
	\[
	\True{p}=U.
	\]
	
	Then
	\[
	\St(u_0)=\{x_1,x_2,x_3\}\subseteq U.
	\]
	Moreover, condition \textbf{(SN)} holds for the designated event $U$ at
	$u_0$: for each facet in $\St(u_0)$ and each other agent $b$, the
	corresponding $b$-column contains a facet outside $U$. Concretely, the
	witnesses can be chosen as
	\[
	z_1 \notin U \text{ for the column of } x_1,\qquad
	y_2 \notin U \text{ for the column of } x_2,\qquad
	y_3 \notin U \text{ for the column of } x_3.
	\]
	
	Hence $U$ is a legitimate secrecy event for agent $a$ at the local state
	$u_0$. Since all facets in $\St(u_0)$ satisfy $p$, we have
	\[
	M,x_i\forces K_ap
	\qquad (i=1,2,3).
	\]
	Because also
	\[
	\True{p}=U\in N_a^S(u_0),
	\]
	it follows that
	\[
	M,x_i\forces S_ap
	\qquad (i=1,2,3).
	\]
	Thus $p$ is secret for agent $a$ throughout the local state $u_0$.
\end{example}

\noindent
\Cref{fig:running-example} makes the geometry of the example explicit. Sharing
the same row means sharing the same $a$-vertex, and sharing the same column
means sharing the same $b$-vertex. The example therefore visualizes at once the
local star $\St(u_0)$, the relations $\sim_a$ and $\sim_b$, and the
external-uncertainty requirement built into secrecy. In
Section~\ref{sec:axiomatic-system} we will reuse the same $3\times 3$ geometry,
with different valuations and secrecy neighborhoods, to show that secrecy is not
closed under self-iteration, conjunction, distribution, or monotonicity.

\subsection{Immediate semantic consequences}

The preceding definitions are designed so that the most basic intended
principles become valid.

\begin{proposition}\label{prop:new-basic}
	Let $M$ be a simplicial secrecy model, $X\in\Fac(M)$, $a\in A$, and
	$\varphi\in\Lang_{KS}$. If
	\[
	M,X\forces S_a\varphi,
	\]
	then:
	\begin{enumerate}
		\item $M,X\forces K_a\varphi$;
		\item $M,X\forces \varphi$;
		\item for every $b\in A\setminus\{a\}$,
		\[
		M,X\forces \neg K_b\varphi;
		\]
		\item for every $b\in A\setminus\{a\}$,
		\[
		M,X\forces \neg K_b\neg\varphi.
		\]
	\end{enumerate}
\end{proposition}

\begin{proof}
	(1) is immediate from the semantic clause for $S_a$.
	
	For (2), if $M,X\forces S_a\varphi$, then by (1) we have
	$M,X\forces K_a\varphi$. Since $\sim_a$ is reflexive by
	\cref{lem:sim-equiv}, it follows that $M,X\forces \varphi$.
	
	For (3), fix $b\neq a$, and let $v=v_a(X)$ and $U=\True{\varphi}$. Since
	$M,X\forces S_a\varphi$, we have $\True{\varphi}\in N_a^S(v)$. Also
	$X\in\St(v)$. By \textbf{(SN)}, there exists $Y\in\Fac(M)$ such that
	\[
	X\sim_b Y
	\qquad\text{and}\qquad
	Y\notin \True{\varphi}.
	\]
	Hence $M,Y\not\forces \varphi$, so $M,X\forces \neg K_b\varphi$.
	
	For (4), by (2) we have $M,X\forces \varphi$, and trivially $X\sim_b X$.
	Hence not all $b$-indistinguishable facets satisfy $\neg\varphi$, and
	therefore
	\[
	M,X\forces \neg K_b\neg\varphi.
	\]
\end{proof}

\begin{corollary}\label{cor:new-basic-truth}
	For every simplicial secrecy model $M$, every facet $X\in\Fac(M)$, every
	agent $a\in A$, and every formula $\varphi\in\Lang_{KS}$,
	\[
	M,X\forces S_a\varphi \rightarrow
	\left(
	K_a\varphi \land \varphi \land
	\bigwedge_{b\in A\setminus\{a\}}
	(\neg K_b\varphi \land \neg K_b\neg\varphi)
	\right).
	\]
\end{corollary}

\begin{proof}
	Immediate from \cref{prop:new-basic}.
\end{proof}

The corollary shows that secrecy, in the present framework, is stronger than
ordinary knowledge and stronger than the mere ignorance of others. If
$\varphi$ is a secret of agent $a$, then $a$ knows $\varphi$, $\varphi$ is
true, and every other agent remains epistemically unable to settle either
$\varphi$ or its negation. This profile aligns with earlier epistemic analyses
of secrecy and related knowledge-based treatments of informational
hiding~\cite{halpern2002secrecy,halpern2005anonymity,xiong2023secrets}.


\section{Axiomatic System}\label{sec:axiomatic-system}

In this section we formulate the proof system corresponding to the semantics from
\cref{sec:syntax-semantics}. The operator $K_a$ retains the familiar S5 behavior
induced by simplicial indistinguishability, while the operator $S_a$ is treated
as a primitive and genuinely non-normal modality. At the same time, because the
truth of $S_a\varphi$ depends only on the owner's local state together with the
truth set of $\varphi$, the semantics validates an exact owner-locality
principle. This is why the axiom
\[
S_a\varphi \rightarrow K_aS_a\varphi
\]
appears naturally in the system.

We call the resulting proof system \emph{Simplicial Secrecy Logic}, abbreviated
by
\[
\mathsf{SSL}.
\]
We write
\[
\vdash_{\mathsf{SSL}} \varphi
\]
to mean that $\varphi$ is derivable in this system.

\subsection{The system \texorpdfstring{$\mathsf{SSL}$}{SSL}}

\subsubsection*{Axiom schemes}

\begin{enumerate}[label=\textbf{(A\arabic*)},leftmargin=2.8em]
	\item All instances of propositional tautologies.
	
	\item \textbf{(K)} \quad
	\[
	K_a(\varphi\rightarrow\psi)\rightarrow (K_a\varphi\rightarrow K_a\psi),
	\qquad a\in A.
	\]
	
	\item \textbf{(T)} \quad
	\[
	K_a\varphi\rightarrow\varphi,
	\qquad a\in A.
	\]
	
	\item \textbf{(4)} \quad
	\[
	K_a\varphi\rightarrow K_aK_a\varphi,
	\qquad a\in A.
	\]
	
	\item \textbf{(5)} \quad
	\[
	\neg K_a\varphi\rightarrow K_a\neg K_a\varphi,
	\qquad a\in A.
	\]
	
	\item \textbf{(S1)} \quad
	\[
	S_a\varphi\rightarrow K_a\varphi,
	\qquad a\in A.
	\]
	
	\item \textbf{(S2)} \quad
	\[
	S_a\varphi\rightarrow \neg K_b\varphi,
	\qquad a,b\in A,\ b\neq a.
	\]
	
	\item \textbf{(S4)} \quad
	\[
	S_a\varphi\rightarrow K_a S_a\varphi,
	\qquad a\in A.
	\]
\end{enumerate}

\subsubsection*{Inference rules}

\begin{enumerate}[label=\textbf{(R\arabic*)},leftmargin=2.8em]
	\item \textbf{Modus Ponens}: from $\varphi$ and $\varphi\rightarrow\psi$, infer
	$\psi$.
	
	\item \textbf{Knowledge Necessitation}: from $\varphi$, infer $K_a\varphi$ for
	every $a\in A$.
	
	\item \textbf{Replacement of Equivalents for Secrecy}: from
	$\vdash_{\mathsf{SSL}} \varphi\leftrightarrow\psi$, infer
	\[
	\vdash_{\mathsf{SSL}} S_a\varphi\leftrightarrow S_a\psi
	\qquad\text{for every } a\in A.
	\]
\end{enumerate}

\begin{remark}
	The operator $S_a$ remains genuinely non-normal: it is extensional, but in
	general it is neither monotone nor closed under conjunction, and it does not
	validate the usual distribution principles of normal modal logics. We will
	make these failures explicit by countermodels below.
\end{remark}

\subsection{Derived principles}

The system $\mathsf{SSL}$ yields a number of useful consequences that clarify the
epistemic profile of secrecy.

\begin{proposition}\label{prop:ssl-truth}
	For every agent $a\in A$,
	\[
	\vdash_{\mathsf{SSL}} S_a\varphi \rightarrow \varphi.
	\]
\end{proposition}

\begin{proof}
	By \textbf{(S1)} we have
	\[
	\vdash_{\mathsf{SSL}} S_a\varphi \rightarrow K_a\varphi.
	\]
	By \textbf{(T)} we also have
	\[
	\vdash_{\mathsf{SSL}} K_a\varphi \rightarrow \varphi.
	\]
	Hence, by propositional reasoning,
	\[
	\vdash_{\mathsf{SSL}} S_a\varphi \rightarrow \varphi.
	\]
\end{proof}

\begin{proposition}[Additional derivable secrecy principles]\label{prop:derived-secrecy-axioms}
	For all distinct agents $a,b\in A$:
	\[
	\vdash_{\mathsf{SSL}} S_a\varphi \rightarrow \neg K_b\neg\varphi,
	\]
	and for every $a\in A$,
	\[
	\vdash_{\mathsf{SSL}} \neg S_a\varphi \rightarrow K_a\neg S_a\varphi.
	\]
\end{proposition}

\begin{proof}
	For the first formula, by \cref{prop:ssl-truth} we have
	\[
	\vdash_{\mathsf{SSL}} S_a\varphi \rightarrow \varphi.
	\]
	Also, by \textbf{(T)} applied to the formula $\neg\varphi$,
	\[
	\vdash_{\mathsf{SSL}} K_b\neg\varphi \rightarrow \neg\varphi.
	\]
	By propositional reasoning,
	\[
	\vdash_{\mathsf{SSL}} \varphi \rightarrow \neg K_b\neg\varphi.
	\]
	Hence
	\[
	\vdash_{\mathsf{SSL}} S_a\varphi \rightarrow \neg K_b\neg\varphi.
	\]
	
	For the second formula, by \textbf{(T)} applied to $S_a\varphi$,
	\[
	\vdash_{\mathsf{SSL}} K_aS_a\varphi \rightarrow S_a\varphi.
	\]
	Hence, by contraposition,
	\[
	\vdash_{\mathsf{SSL}} \neg S_a\varphi \rightarrow \neg K_aS_a\varphi.
	\]
	By \textbf{(5)} applied to the formula $S_a\varphi$,
	\[
	\vdash_{\mathsf{SSL}} \neg K_aS_a\varphi \rightarrow K_a\neg K_aS_a\varphi.
	\]
	Also, by contraposition of \textbf{(S4)},
	\[
	\vdash_{\mathsf{SSL}} \neg K_aS_a\varphi \rightarrow \neg S_a\varphi.
	\]
	By Knowledge Necessitation and \textbf{(K)},
	\[
	\vdash_{\mathsf{SSL}} K_a\neg K_aS_a\varphi \rightarrow K_a\neg S_a\varphi.
	\]
	Combining these implications propositionally, we obtain
	\[
	\vdash_{\mathsf{SSL}} \neg S_a\varphi \rightarrow K_a\neg S_a\varphi.
	\]
\end{proof}

\begin{corollary}[Exact owner-locality]\label{cor:exact-owner-locality}
	For every agent $a\in A$,
	\[
	\vdash_{\mathsf{SSL}} S_a\varphi \leftrightarrow K_aS_a\varphi,
	\qquad
	\vdash_{\mathsf{SSL}} \neg S_a\varphi \leftrightarrow K_a\neg S_a\varphi.
	\]
\end{corollary}

\begin{proof}
	The left equivalence follows from \textbf{(S4)} together with \textbf{(T)}
	applied to $S_a\varphi$. The right equivalence follows from
	\cref{prop:derived-secrecy-axioms} together with \textbf{(T)} applied to
	$\neg S_a\varphi$.
\end{proof}

\begin{proposition}\label{prop:ssl-summary}
	For every agent $a\in A$,
	\[
	\vdash_{\mathsf{SSL}}
	S_a\varphi \rightarrow
	\left(
	K_a\varphi \land \varphi \land
	\bigwedge_{b\in A\setminus\{a\}}(\neg K_b\varphi \land \neg K_b\neg\varphi)
	\right).
	\]
\end{proposition}

\begin{proof}
	Combine \textbf{(S1)}, \textbf{(S2)}, \cref{prop:ssl-truth}, and
	\cref{prop:derived-secrecy-axioms}.
\end{proof}

\begin{proposition}[Owner knows others' ignorance]\label{prop:owner-knows-ignorance}
	For all distinct agents $a,b\in A$,
	\[
	\vdash_{\mathsf{SSL}} S_a\varphi \rightarrow K_a\neg K_b\varphi,
	\qquad
	\vdash_{\mathsf{SSL}} S_a\varphi \rightarrow K_a\neg K_b\neg\varphi.
	\]
\end{proposition}

\begin{proof}
	By \textbf{(S2)},
	\[
	\vdash_{\mathsf{SSL}} S_a\varphi \rightarrow \neg K_b\varphi.
	\]
	By Knowledge Necessitation and \textbf{(K)},
	\[
	\vdash_{\mathsf{SSL}} K_aS_a\varphi \rightarrow K_a\neg K_b\varphi.
	\]
	Using \textbf{(S4)}, we obtain
	\[
	\vdash_{\mathsf{SSL}} S_a\varphi \rightarrow K_a\neg K_b\varphi.
	\]
	The second implication is identical, using
	\cref{prop:derived-secrecy-axioms} in place of \textbf{(S2)}.
\end{proof}

\begin{proposition}[Higher-order opacity]\label{prop:higher-order-opacity}
	For all distinct agents $a,b\in A$,
	\[
	\vdash_{\mathsf{SSL}} S_a\varphi \rightarrow \neg K_b S_a\varphi,
	\qquad
	\vdash_{\mathsf{SSL}} S_a\varphi \rightarrow \neg K_b\neg S_a\varphi.
	\]
\end{proposition}

\begin{proof}
	By \cref{prop:ssl-truth},
	\[
	\vdash_{\mathsf{SSL}} S_a\varphi \rightarrow \varphi.
	\]
	By Knowledge Necessitation and \textbf{(K)},
	\[
	\vdash_{\mathsf{SSL}} K_bS_a\varphi \rightarrow K_b\varphi.
	\]
	Combining this with \textbf{(S2)}, we obtain
	\[
	\vdash_{\mathsf{SSL}} S_a\varphi \rightarrow \neg K_bS_a\varphi.
	\]
	
	For the second implication, by \textbf{(T)} applied to the formula
	$\neg S_a\varphi$,
	\[
	\vdash_{\mathsf{SSL}} K_b\neg S_a\varphi \rightarrow \neg S_a\varphi.
	\]
	Hence, by propositional reasoning,
	\[
	\vdash_{\mathsf{SSL}} S_a\varphi \rightarrow \neg K_b\neg S_a\varphi.
	\]
\end{proof}

\begin{proposition}[No foreign secrets about $b$-local facts]\label{prop:no-foreign-local-secrets}
	If
	\[
	\vdash_{\mathsf{SSL}} \chi \rightarrow K_b\chi,
	\]
	then for every $a\neq b$,
	\[
	\vdash_{\mathsf{SSL}} \neg S_a\chi.
	\]
	In particular, for all distinct $a,b\in A$,
	\[
	\vdash_{\mathsf{SSL}} \neg S_a K_b\psi,\qquad
	\vdash_{\mathsf{SSL}} \neg S_a \neg K_b\psi,\qquad
	\vdash_{\mathsf{SSL}} \neg S_a S_b\psi,\qquad
	\vdash_{\mathsf{SSL}} \neg S_a \neg S_b\psi.
	\]
\end{proposition}

\begin{proof}
	Assume $\vdash_{\mathsf{SSL}} \chi \rightarrow K_b\chi$ and $a\neq b$. By
	\cref{prop:ssl-truth},
	\[
	\vdash_{\mathsf{SSL}} S_a\chi \rightarrow \chi.
	\]
	Hence
	\[
	\vdash_{\mathsf{SSL}} S_a\chi \rightarrow K_b\chi.
	\]
	By \textbf{(S2)},
	\[
	\vdash_{\mathsf{SSL}} S_a\chi \rightarrow \neg K_b\chi.
	\]
	So, by propositional reasoning,
	\[
	\vdash_{\mathsf{SSL}} \neg S_a\chi.
	\]
	
	The displayed instances follow respectively from \textbf{(4)}, \textbf{(5)},
	\textbf{(S4)}, and \cref{prop:derived-secrecy-axioms}.
\end{proof}

\begin{proposition}\label{prop:no-secret-top}
	Assume $|A|\ge 2$. Then for every $a\in A$,
	\[
	\vdash_{\mathsf{SSL}} \neg S_a\top.
	\]
\end{proposition}

\begin{proof}
	Fix $a\in A$ and choose $b\neq a$. Since $\top$ is a theorem, so is $K_b\top$.
	But by \textbf{(S2)},
	\[
	\vdash_{\mathsf{SSL}} S_a\top \rightarrow \neg K_b\top.
	\]
	Hence
	\[
	\vdash_{\mathsf{SSL}} \neg S_a\top.
	\]
\end{proof}

These results show that the system already derives a rich epistemic profile for
secrecy. In particular, secrecy implies truth, exact owner-locality, and a
strong form of higher-order opacity, while ruling out foreign secrets about facts
that are already local to another agent. This is precisely the sort of
proof-theoretic behavior one would expect from the semantics introduced in
\cref{sec:syntax-semantics}.

\subsection{Non-normality: countermodels}

We now show that, despite its owner-locality, the operator $S_a$ remains
genuinely non-normal. To make the comparison transparent, we keep fixed the
underlying $3\times 3$ geometry from the running example in
\cref{subsec:running-example} and vary only the valuation and the secrecy
neighborhoods.

\begin{example}[Failure of secrecy idempotence]\label{ex:no-pos-intro}
	Let $A=\{a,b\}$. Consider the pure $A$-chromatic simplicial complex with
	vertices
	\[
	u_0,u_1,u_2 \text{ of colour } a,
	\qquad
	w_1,w_2,w_3 \text{ of colour } b,
	\]
	and facets
	\[
	x_i=\{u_0,w_i\},\qquad
	y_i=\{u_1,w_i\},\qquad
	z_i=\{u_2,w_i\}
	\qquad (i=1,2,3).
	\]
	
	Let $p$ be true exactly on
	\[
	\True{p}=\{x_1,x_2,x_3,y_1\},
	\]
	let
	\[
	N_a^S(u_0)=\{\True{p}\},
	\]
	and let all other secrecy neighborhoods be empty.
	
	The condition \textbf{(SN)} is satisfied: for each $x_i\in\St(u_0)$, the
	$b$-indistinguishable column of $x_i$ contains a facet outside $\True{p}$
	(namely $z_1$ for $x_1$, and $y_i$ or $z_i$ for $i=2,3$).
	
	Now, at each $x_i$, we have $K_ap$ and $\True{p}\in N_a^S(u_0)$, so
	\[
	M,x_i\forces S_ap.
	\]
	But $\True{S_ap}=\{x_1,x_2,x_3\}$, and this set is not in $N_a^S(u_0)$.
	Hence
	\[
	M,x_i\not\forces S_aS_ap.
	\]
	Therefore
	\[
	\not\models S_a\varphi\rightarrow S_aS_a\varphi.
	\]
\end{example}

\begin{example}[Failure of conjunction closure, distribution, and monotonicity]\label{ex:no-distribution}
	On the same underlying simplicial complex, let propositional variables
	$p,q,r$ be interpreted by
	\[
	\True{p}=\{x_1,x_2,x_3,y_1,y_2,y_3\},
	\]
	\[
	\True{q}=\{x_1,x_2,x_3\},
	\]
	\[
	\True{r}=\{x_1,x_2,x_3,z_1,z_2,z_3\},
	\]
	and let
	\[
	N_a^S(u_0)=\{\True{p},\True{r}\},
	\]
	with all other secrecy neighborhoods empty.
	
	Again \textbf{(SN)} is satisfied: each $b$-column contains a witness outside
	$\True{p}$ and also a witness outside $\True{r}$.
	
	At each $x_i$, both $p$ and $r$ are known by $a$, and their truth sets are
	designated at $u_0$. Hence
	\[
	M,x_i\forces S_ap\land S_ar.
	\]
	But
	\[
	\True{p\land r}=\{x_1,x_2,x_3\}=\True{q},
	\]
	and this set is not in $N_a^S(u_0)$. Therefore
	\[
	M,x_i\not\forces S_a(p\land r),
	\]
	so
	\[
	\not\models (S_a\varphi\land S_a\psi)\rightarrow S_a(\varphi\land\psi).
	\]
	
	Moreover,
	\[
	\True{p\rightarrow q}=\True{r}.
	\]
	So at each $x_i$,
	\[
	M,x_i\forces S_a(p\rightarrow q)\land S_ap,
	\]
	but
	\[
	M,x_i\not\forces S_aq.
	\]
	Hence
	\[
	\not\models S_a(\varphi\rightarrow\psi)\rightarrow (S_a\varphi\rightarrow S_a\psi).
	\]
	
	The same model also witnesses failure of monotonicity. Since
	\[
	\models p\rightarrow (p\lor r),
	\]
	monotonicity of $S_a$ would imply
	\[
	\models S_ap\rightarrow S_a(p\lor r).
	\]
	But
	\[
	\True{p\lor r}=\Fac(M),
	\]
	and
	\[
	\Fac(M)\notin N_a^S(u_0).
	\]
	Hence at each $x_i$,
	\[
	M,x_i\forces S_ap
	\qquad\text{but}\qquad
	M,x_i\not\forces S_a(p\lor r).
	\]
	Therefore
	\[
	\not\models S_ap\rightarrow S_a(p\lor r),
	\]
	so $S_a$ is not monotone.
\end{example}

These countermodels show that owner-locality should not be mistaken for
normality. Even though secrecy is constant across the owner's epistemic range,
it is not closed under self-iteration, conjunction, material implication, or
monotone consequence.

\subsection{Proof-theoretic stance}

For the genuinely multi-agent case $|A|\ge 2$, which is the setting of the
completeness theorem in \cref{sec:metatheory}, the system $\mathsf{SSL}$
should be understood as the exact axiomatic system for the current
vertex-based secrecy semantics. It incorporates owner-locality because that
principle is already forced by the dependence of $S_a$ on the owner's local
state. At the same time, the countermodels above show that secrecy remains
genuinely non-normal: one should not add stronger normal-style distribution
or closure principles without imposing further restrictions on the secrecy
neighborhoods.

In this sense, the proof theory mirrors the semantics exactly. The system is
strong enough to derive the intended epistemic consequences of secrecy, but it
stops short of imposing normal modal behavior that the semantics does not
support. In the next section we show that this system is sound with respect to
simplicial secrecy models.


\section{Soundness and Completeness}\label{sec:metatheory}

In this section we establish the metatheory of the system $\mathsf{SSL}$
introduced in \cref{sec:axiomatic-system}. We first prove soundness with respect
to the class of simplicial secrecy models from \cref{sec:syntax-semantics}. We
then prove completeness, for the genuinely multi-agent case $|A|\ge 2$, by
combining an auxiliary-colour canonical construction with a representation
theorem back into the original class of pure $A$-chromatic simplicial secrecy
models.

\subsection{Soundness}

We begin with the standard semantic notion of validity.

\begin{definition}[Validity]
	Let $M$ be a simplicial secrecy model and let $\varphi\in\Lang_{KS}$.
	\begin{enumerate}
		\item We write
		\[
		M \models \varphi
		\]
		if for every facet $X\in\Fac(M)$,
		\[
		M,X\forces \varphi.
		\]
		\item We write
		\[
		\models \varphi
		\]
		if for every simplicial secrecy model $M$,
		\[
		M\models \varphi.
		\]
	\end{enumerate}
\end{definition}

\begin{proposition}[Semantic owner-locality]\label{prop:owner-locality-valid}
	For every agent $a\in A$,
	\[
	\models S_a\varphi \leftrightarrow K_aS_a\varphi.
	\]
	Consequently,
	\[
	\models S_a\varphi \rightarrow K_aS_a\varphi,
	\qquad
	\models \neg S_a\varphi \rightarrow K_a\neg S_a\varphi.
	\]
\end{proposition}

\begin{proof}
	Let $M$ be a simplicial secrecy model and let $X\in\Fac(M)$.
	
	We first show that for every facet $Y\in\Fac(M)$ with
	\[
	X\sim_a Y,
	\]
	we have
	\[
	M,X\forces S_a\varphi
	\quad\Longleftrightarrow\quad
	M,Y\forces S_a\varphi.
	\]
	
	So assume $X\sim_a Y$, and write
	\[
	v:=v_a(X)=v_a(Y).
	\]
	Then for every facet $Z\in\Fac(M)$,
	\[
	X\sim_a Z
	\quad\Longleftrightarrow\quad
	v_a(Z)=v
	\quad\Longleftrightarrow\quad
	Y\sim_a Z.
	\]
	Hence
	\[
	M,X\forces K_a\varphi
	\quad\Longleftrightarrow\quad
	M,Y\forces K_a\varphi.
	\]
	Moreover,
	\[
	N_a^S(v_a(X))=N_a^S(v)=N_a^S(v_a(Y)),
	\]
	and the truth set $\True{\varphi}$ is independent of the evaluation point.
	Therefore
	\[
	M,X\forces S_a\varphi
	\quad\Longleftrightarrow\quad
	M,Y\forces S_a\varphi.
	\]
	So $S_a\varphi$ is constant on each $a$-equivalence class.
	
	Now suppose
	\[
	M,X\forces S_a\varphi.
	\]
	Then every facet $Y$ with $X\sim_a Y$ also satisfies $S_a\varphi$, and hence
	\[
	M,X\forces K_aS_a\varphi.
	\]
	
	Conversely, suppose
	\[
	M,X\forces K_aS_a\varphi.
	\]
	Since $\sim_a$ is reflexive, we have $X\sim_a X$, and therefore
	\[
	M,X\forces S_a\varphi.
	\]
	
	Thus
	\[
	M,X\forces S_a\varphi
	\quad\Longleftrightarrow\quad
	M,X\forces K_aS_a\varphi.
	\]
	Since $M$ and $X$ were arbitrary,
	\[
	\models S_a\varphi \leftrightarrow K_aS_a\varphi.
	\]
	
	Finally, suppose
	\[
	M,X\forces \neg S_a\varphi.
	\]
	By the class-constancy proved above, every facet $Y$ with $X\sim_a Y$
	also satisfies $\neg S_a\varphi$. Hence
	\[
	M,X\forces K_a\neg S_a\varphi.
	\]
	Since $M$ and $X$ were arbitrary,
	\[
	\models \neg S_a\varphi \rightarrow K_a\neg S_a\varphi.
	\]
	
	The first displayed consequence now follows immediately from
	\[
	\models S_a\varphi \leftrightarrow K_aS_a\varphi.
	\]
\end{proof}

\begin{theorem}[Soundness of $\mathsf{SSL}$]\label{thm:ssl-sound}
	For every formula $\varphi\in\Lang_{KS}$,
	\[
	\vdash_{\mathsf{SSL}} \varphi
	\quad\Longrightarrow\quad
	\models \varphi.
	\]
\end{theorem}

\begin{proof}
	We show that every axiom scheme of $\mathsf{SSL}$ is valid on simplicial
	secrecy models and that every inference rule preserves validity. The theorem
	then follows by induction on the length of derivations.
	
The proofs for propositional tautologies and for the S5-part of $K_a$
(\textbf{(K)}, \textbf{(T)}, \textbf{(4)}, \textbf{(5)}) are exactly as in
the ordinary simplicial epistemic setting, since by
\cref{lem:sim-equiv} each relation $\sim_a$ is an equivalence relation on
facets.

	For the secrecy axioms:
	\begin{itemize}
		\item \textbf{(S1)} is valid by \cref{prop:new-basic}(1);
		\item \textbf{(S2)} is valid by \cref{prop:new-basic}(3);
		\item \textbf{(S4)} is valid by \cref{prop:owner-locality-valid}.
	\end{itemize}
	
	Modus Ponens and Knowledge Necessitation preserve validity exactly as usual.
	Replacement of Equivalents for Secrecy preserves validity because the truth of
	$S_a\varphi$ depends only on the truth set $\True{\varphi}$ together with the
	ordinary knowledge condition, both of which are preserved under valid
	equivalence.
\end{proof}

\begin{corollary}\label{cor:ssl-derived-sound}
	Every formula derivable in $\mathsf{SSL}$ is valid on the class of simplicial
	secrecy models. In particular, the derived principles from
	\cref{sec:axiomatic-system}, such as
	\[
	S_a\varphi\rightarrow\varphi,
	\qquad
	S_a\varphi\rightarrow
	\bigwedge_{b\in A\setminus\{a\}}
	(\neg K_b\varphi \land \neg K_b\neg\varphi),
	\]
	are semantically valid.
\end{corollary}

\begin{remark}
	The proof of soundness highlights the division of labour between the two
	semantic layers of our framework. The knowledge axioms are validated by the
	ordinary simplicial epistemic base, namely by the fact that the relations
	$\sim_a$ are equivalence relations on facets, as in standard simplicial
	epistemic semantics~\cite{goubault2021simplicial,vanditmarsch2022knowledge}.
	By contrast, the secrecy principles are validated by three ingredients
	working together: the explicit knowledge conjunct in the truth clause for
	$S_a$, the external-uncertainty frame condition \textbf{(SN)}, and the fact
	that secrecy depends only on the owner's local state. This confirms the
	intended methodological picture: secrecy is not reducible to knowledge, but
	is added as a genuinely new semantic layer on top of the standard simplicial
	interpretation of epistemic logic.
\end{remark}

\subsection{Canonical preliminaries}

We now turn to completeness. The canonical construction below uses only finitary
proof-theoretic machinery.

\begin{definition}[Consistency]
	A set $\Sigma\subseteq\Lang_{KS}$ is \emph{$\mathsf{SSL}$-consistent} if there
	are no formulas $\phi_1,\dots,\phi_n\in\Sigma$ such that
	\[
	\vdash_{\mathsf{SSL}}(\phi_1\land\dots\land\phi_n)\rightarrow\bot,
	\]
	where $n\ge 0$ and the empty conjunction is understood as $\top$.
\end{definition}

By the standard Lindenbaum extension lemma for Hilbert systems with finitary
rules, every $\mathsf{SSL}$-consistent set of formulas can be extended to a
maximally $\mathsf{SSL}$-consistent set; compare the standard canonical-model
machinery for modal and neighborhood logics~\cite{chellas1980modal}.

\paragraph{Standing assumption.}
Henceforth, unless stated otherwise, we assume
\[
|A|\ge 2.
\]
The role of this assumption will be discussed explicitly at the end of the
section; see \cref{rem:one-agent-role}.

Let $\mathsf{MCS}$ be the set of all maximally $\mathsf{SSL}$-consistent sets of
formulas.

For $\Gamma,\Delta\in\mathsf{MCS}$ and $a\in A$, define
\[
\Gamma R_a \Delta
\quad\text{iff}\quad
\{\varphi \mid K_a\varphi\in\Gamma\}\subseteq \Delta.
\]

\begin{lemma}\label{lem:Ra-equivalence}
	For every agent $a\in A$, the relation $R_a$ on $\mathsf{MCS}$ is an
	equivalence relation.
\end{lemma}

\begin{proof}
	We first show that $R_a$ is reflexive.
	
	Let $\Gamma\in\mathsf{MCS}$, and suppose $K_a\varphi\in\Gamma$. By axiom
	\textbf{(T)},
	\[
	K_a\varphi\rightarrow \varphi
	\]
	is a theorem of $\mathsf{SSL}$, hence belongs to $\Gamma$. Since $\Gamma$ is
	closed under modus ponens, it follows that $\varphi\in\Gamma$. Therefore
	$\Gamma R_a\Gamma$.
	
	Next we show that $R_a$ is Euclidean.
	
	Assume $\Gamma R_a\Delta$ and $\Gamma R_a\Theta$. We show that
	$\Delta R_a\Theta$. Let $K_a\varphi\in\Delta$. Suppose, for contradiction,
	that $\varphi\notin\Theta$. Then $\neg\varphi\in\Theta$ by maximal
	consistency.
	
	We claim that $K_a\varphi\in\Gamma$. If not, then $\neg K_a\varphi\in\Gamma$.
	By axiom \textbf{(5)},
	\[
	\neg K_a\varphi \rightarrow K_a\neg K_a\varphi
	\]
	belongs to $\Gamma$, so by modus ponens
	\[
	K_a\neg K_a\varphi\in\Gamma.
	\]
	Since $\Gamma R_a\Delta$, it follows that
	\[
	\neg K_a\varphi\in\Delta,
	\]
	contradicting $K_a\varphi\in\Delta$. Hence indeed $K_a\varphi\in\Gamma$.
	
	But then, since $\Gamma R_a\Theta$, we obtain $\varphi\in\Theta$,
	contradicting $\neg\varphi\in\Theta$. Therefore $\varphi\in\Theta$, and so
	$\Delta R_a\Theta$.
	
	Thus $R_a$ is Euclidean.
	
	Now symmetry follows from reflexivity and Euclideanness. Suppose
	$\Gamma R_a\Delta$. Since $\Gamma R_a\Gamma$ and $R_a$ is Euclidean, it
	follows that $\Delta R_a\Gamma$.
	
	Finally, transitivity follows from symmetry and Euclideanness. Suppose
	$\Gamma R_a\Delta$ and $\Delta R_a\Theta$. By symmetry, from
	$\Gamma R_a\Delta$ we obtain $\Delta R_a\Gamma$. Since $R_a$ is Euclidean,
	from $\Delta R_a\Gamma$ and $\Delta R_a\Theta$ we obtain $\Gamma R_a\Theta$.
	
	Therefore $R_a$ is an equivalence relation.
\end{proof}

For $\Gamma\in\mathsf{MCS}$ and $a\in A$, let
\[
[\Gamma]_a:=\{\Delta\in\mathsf{MCS}\mid \Gamma R_a\Delta\}.
\]

\begin{lemma}[Derivable conjunction for knowledge]\label{lem:K-conj}
	For every agent $a\in A$ and formulas $\psi_1,\dots,\psi_n$ with $n\ge 1$,
	\[
	\vdash_{\mathsf{SSL}}
	(K_a\psi_1\land\cdots\land K_a\psi_n)\rightarrow
	K_a(\psi_1\land\cdots\land\psi_n).
	\]
\end{lemma}

\begin{proof}
	By induction on $n$. The case $n=1$ is immediate.
	
	For $n=2$, use the propositional tautology
	\[
	\psi_1\rightarrow(\psi_2\rightarrow(\psi_1\land\psi_2)).
	\]
	By Knowledge Necessitation,
	\[
	\vdash_{\mathsf{SSL}}
	K_a\bigl(\psi_1\rightarrow(\psi_2\rightarrow(\psi_1\land\psi_2))\bigr).
	\]
	By two applications of \textbf{(K)},
	\[
	\vdash_{\mathsf{SSL}}
	K_a\psi_1\rightarrow\bigl(K_a\psi_2\rightarrow K_a(\psi_1\land\psi_2)\bigr).
	\]
	So
	\[
	\vdash_{\mathsf{SSL}}
	(K_a\psi_1\land K_a\psi_2)\rightarrow K_a(\psi_1\land\psi_2).
	\]
	The induction step is routine.
\end{proof}

\begin{lemma}[Canonical existence lemma]\label{lem:canonical-existence}
	If $\neg K_a\varphi\in\Gamma$, then there exists $\Delta\in\mathsf{MCS}$ such
	that
	\[
	\Gamma R_a \Delta
	\qquad\text{and}\qquad
	\neg\varphi\in\Delta.
	\]
\end{lemma}

\begin{proof}
	Consider the set
	\[
	\Sigma := \{\psi\mid K_a\psi\in\Gamma\}\cup\{\neg\varphi\}.
	\]
	We claim that $\Sigma$ is $\mathsf{SSL}$-consistent.
	
	Suppose not. Then some finite subset of $\Sigma$ is inconsistent. So there
	exist formulas $\psi_1,\dots,\psi_n$ such that $K_a\psi_i\in\Gamma$ for all
	$i=1,\dots,n$ and
	\[
	\vdash_{\mathsf{SSL}} (\psi_1\land\dots\land\psi_n)\rightarrow\varphi,
	\]
	where $n\ge 0$.
	
	If $n=0$, then simply
	\[
	\vdash_{\mathsf{SSL}} \varphi.
	\]
	By Knowledge Necessitation,
	\[
	\vdash_{\mathsf{SSL}} K_a\varphi,
	\]
	contradicting $\neg K_a\varphi\in\Gamma$.
	
	So assume $n\ge 1$. By Knowledge Necessitation,
	\[
	\vdash_{\mathsf{SSL}}
	K_a\bigl((\psi_1\land\dots\land\psi_n)\rightarrow\varphi\bigr).
	\]
	By \textbf{(K)},
	\[
	\vdash_{\mathsf{SSL}}
	K_a(\psi_1\land\dots\land\psi_n)\rightarrow K_a\varphi.
	\]
	By \cref{lem:K-conj}, from $K_a\psi_1,\dots,K_a\psi_n\in\Gamma$ we obtain
	\[
	K_a(\psi_1\land\dots\land\psi_n)\in\Gamma.
	\]
	Hence $K_a\varphi\in\Gamma$, contradicting $\neg K_a\varphi\in\Gamma$.
	
	Therefore $\Sigma$ is consistent. Extend it to some $\Delta\in\mathsf{MCS}$.
	Then $\Gamma R_a\Delta$ by construction, and $\neg\varphi\in\Delta$.
\end{proof}

\subsection{Auxiliary-colour canonical model}

We first introduce a temporary auxiliary colour.

\begin{definition}[Auxiliary-colour simplicial secrecy model]\label{def:aux-ssm}
	Fix a fresh colour $*\notin A$. An \emph{auxiliary-colour simplicial secrecy
		model} (over $A$) is a tuple
	\[
	M^\ast=(V,\mathcal{F},\chi,\nu,\{N_a^S\}_{a\in A})
	\]
	such that:
	\begin{enumerate}
		\item $(V,\mathcal{F},\chi)$ is an $(A\cup\{*\})$-chromatic simplicial
		complex;
		\item every facet $X\in\Fac(M^\ast)$ satisfies
		\[
		\chi[X]=A\cup\{*\};
		\]
		\item every vertex belongs to some facet of $M^\ast$, i.e.
		\[
		V=\bigcup_{X\in\Fac(M^\ast)}X;
		\]
		\item $\nu:\Fac(M^\ast)\to\Pow(\Prop)$ is a valuation function;
		\item for each agent $a\in A$, if
		\[
		V_a=\{v\in V\mid \chi(v)=a\},
		\]
		then
		\[
		N_a^S:V_a\to\Pow\bigl(\Pow(\Fac(M^\ast))\bigr)
		\]
		satisfies \textbf{(SN)} exactly as in \cref{def:ssm}.
	\end{enumerate}
	
	For each $a\in A$ and each facet $X\in\Fac(M^\ast)$, we write $v_a(X)$ for
	the unique vertex of colour $a$ in $X$. Formulas of $\Lang_{KS}$ are
	evaluated on auxiliary-colour models exactly as in
	\cref{def:ssm-semantics}; the auxiliary colour $*$ does not occur in the
	language.
\end{definition}

\begin{definition}[Auxiliary canonical model]\label{def:aux-canonical-model}
	The \emph{auxiliary canonical model} is the structure
	\[
	M^{\ast c}=
	(V^{\ast c},\mathcal{F}^{\ast c},\chi^{\ast c},\nu^{\ast c},
	\{N_a^{S,\ast c}\}_{a\in A})
	\]
	defined as follows.
	\begin{enumerate}
		\item The set of vertices is
		\[
		V^{\ast c}
		=
		\{(*,\Gamma)\mid \Gamma\in\mathsf{MCS}\}
		\ \cup\
		\{(a,[\Gamma]_a)\mid \Gamma\in\mathsf{MCS},\ a\in A\}.
		\]
		
		\item For each $\Gamma\in\mathsf{MCS}$, define the facet
		\[
		X_\Gamma^\ast := \{(*,\Gamma)\}\cup\{(a,[\Gamma]_a)\mid a\in A\}.
		\]
		The set of facets is
		\[
		\Fac(M^{\ast c})=\{X_\Gamma^\ast\mid \Gamma\in\mathsf{MCS}\},
		\]
		and $\mathcal{F}^{\ast c}$ is the downward closure of
		$\Fac(M^{\ast c})$.
		
		\item The colouring function is given by
		\[
		\chi^{\ast c}((*,\Gamma))=*,\qquad
		\chi^{\ast c}((a,[\Gamma]_a))=a.
		\]
		
		\item The valuation on facets is given by
		\[
		\nu^{\ast c}(X_\Gamma^\ast)=\{p\in\Prop\mid p\in\Gamma\}.
		\]
		
		\item For each vertex $v=(a,[\Gamma]_a)$ of colour $a\in A$, define
		\[
		N_a^{S,\ast c}(v)
		:=
		\{\widehat{\varphi}^{\,\ast}\subseteq\Fac(M^{\ast c})
		\mid S_a\varphi\in\Gamma\},
		\]
		where
		\[
		\widehat{\varphi}^{\,\ast}:=
		\{X_\Delta^\ast\in\Fac(M^{\ast c})\mid \varphi\in\Delta\}.
		\]
	\end{enumerate}
\end{definition}

\begin{lemma}[Local invariance of canonical secrecy formulas]\label{lem:well-defined-N}
	If $[\Gamma]_a=[\Delta]_a$, then for every formula $\varphi$,
	\[
	S_a\varphi\in\Gamma
	\quad\Longleftrightarrow\quad
	S_a\varphi\in\Delta.
	\]
	Consequently,
	\[
	N_a^{S,\ast c}((a,[\Gamma]_a))
	=
	N_a^{S,\ast c}((a,[\Delta]_a)).
	\]
\end{lemma}

\begin{proof}
	Assume $[\Gamma]_a=[\Delta]_a$, i.e.\ $\Gamma R_a\Delta$ and
	$\Delta R_a\Gamma$.
	
	Let $\varphi$ be arbitrary.
	
	If $S_a\varphi\in\Gamma$, then by axiom \textbf{(S4)} we have
	\[
	K_aS_a\varphi\in\Gamma.
	\]
	Since $\Gamma R_a\Delta$, it follows that
	\[
	S_a\varphi\in\Delta.
	\]
	
	If $S_a\varphi\notin\Gamma$, then $\neg S_a\varphi\in\Gamma$ by maximal
	consistency. By \cref{prop:derived-secrecy-axioms},
	\[
	\vdash_{\mathsf{SSL}} \neg S_a\varphi \rightarrow K_a\neg S_a\varphi.
	\]
	Hence
	\[
	K_a\neg S_a\varphi\in\Gamma.
	\]
	Since $\Gamma R_a\Delta$, we obtain
	\[
	\neg S_a\varphi\in\Delta,
	\]
	and hence $S_a\varphi\notin\Delta$.
	
	By symmetry, the converse implications also hold. Therefore
	\[
	S_a\varphi\in\Gamma
	\quad\Longleftrightarrow\quad
	S_a\varphi\in\Delta.
	\]
	The equality of the neighborhoods follows immediately from
	\cref{def:aux-canonical-model}.
\end{proof}

\begin{lemma}[Canonical accessibility]\label{lem:aux-canonical-accessibility}
	For all $\Gamma,\Delta\in\mathsf{MCS}$ and all $a\in A$,
	\[
	X_\Gamma^\ast \sim_a X_\Delta^\ast
	\quad\Longleftrightarrow\quad
	\Gamma R_a\Delta.
	\]
\end{lemma}

\begin{proof}
	By definition,
	\[
	X_\Gamma^\ast \sim_a X_\Delta^\ast
	\quad\Longleftrightarrow\quad
	(a,[\Gamma]_a)=(a,[\Delta]_a),
	\]
	which is equivalent to
	\[
	[\Gamma]_a=[\Delta]_a.
	\]
	Since $R_a$ is an equivalence relation by \cref{lem:Ra-equivalence}, this
	holds iff $\Gamma R_a\Delta$.
\end{proof}

\begin{lemma}\label{lem:aux-canonical-ssm}
	The auxiliary canonical structure $M^{\ast c}$ is an auxiliary-colour
	simplicial secrecy model.
\end{lemma}

\begin{proof}
	The underlying simplicial part is immediate: every facet $X_\Gamma^\ast$
	contains exactly one vertex of each colour in $A\cup\{*\}$, so
	$(V^{\ast c},\mathcal{F}^{\ast c},\chi^{\ast c},\nu^{\ast c})$ is an
	$(A\cup\{*\})$-chromatic simplicial epistemic model.
	
	It remains only to verify \textbf{(SN)}.
	
	Let $v=(a,[\Gamma]_a)\in V^{\ast c}$ and let $U\in N_a^{S,\ast c}(v)$. By
	\cref{def:aux-canonical-model}, there exists a formula $\varphi$ such that
	\[
	U=\widehat{\varphi}^{\,\ast}
	\qquad\text{and}\qquad
	S_a\varphi\in\Gamma.
	\]
	
	Let $X_\Delta^\ast\in\St(v)$ and let $b\in A\setminus\{a\}$. Since
	$[\Delta]_a=[\Gamma]_a$, by \cref{lem:well-defined-N} we have
	\[
	S_a\varphi\in\Delta.
	\]
	By axiom \textbf{(S2)},
	\[
	\neg K_b\varphi\in\Delta.
	\]
	By \cref{lem:canonical-existence}, there exists some
	$\Theta\in\mathsf{MCS}$ such that
	\[
	\Delta R_b\Theta
	\qquad\text{and}\qquad
	\neg\varphi\in\Theta.
	\]
	By \cref{lem:aux-canonical-accessibility},
	\[
	X_\Delta^\ast \sim_b X_\Theta^\ast.
	\]
	Since $\neg\varphi\in\Theta$, we have $\varphi\notin\Theta$, hence
	\[
	X_\Theta^\ast\notin \widehat{\varphi}^{\,\ast}=U.
	\]
	Thus \textbf{(SN)} is satisfied.
\end{proof}

\subsection{Truth lemma}

\begin{lemma}[Extensionality lemma]\label{lem:extensionality}
	If
	\[
	\widehat{\varphi}^{\,\ast}=\widehat{\psi}^{\,\ast},
	\]
	then
	\[
	\vdash_{\mathsf{SSL}}\varphi\leftrightarrow\psi.
	\]
\end{lemma}

\begin{proof}
	Assume for contradiction that
	\[
	\not\vdash_{\mathsf{SSL}}\varphi\leftrightarrow\psi.
	\]
	Then
	\[
	\neg(\varphi\leftrightarrow\psi)
	\]
	is $\mathsf{SSL}$-consistent, and can be extended to some
	$\Gamma\in\mathsf{MCS}$.
	
	Since $\Gamma$ is maximally consistent, for each of $\varphi,\psi$, either
	the formula or its negation belongs to $\Gamma$.
	
	The cases $\varphi,\psi\in\Gamma$ and $\neg\varphi,\neg\psi\in\Gamma$ are
	impossible, because the propositional tautologies
	\[
	\varphi\rightarrow(\psi\rightarrow(\varphi\leftrightarrow\psi))
	\]
	and
	\[
	\neg\varphi\rightarrow(\neg\psi\rightarrow(\varphi\leftrightarrow\psi))
	\]
	would then imply $\varphi\leftrightarrow\psi\in\Gamma$, contradicting
	$\neg(\varphi\leftrightarrow\psi)\in\Gamma$.
	
	Hence exactly one of the following holds:
	\[
	\varphi\in\Gamma \text{ and } \neg\psi\in\Gamma,
	\qquad\text{or}\qquad
	\psi\in\Gamma \text{ and } \neg\varphi\in\Gamma.
	\]
	In the first case,
	\[
	X_\Gamma^\ast\in\widehat{\varphi}^{\,\ast}
	\qquad\text{but}\qquad
	X_\Gamma^\ast\notin\widehat{\psi}^{\,\ast},
	\]
	and in the second case,
	\[
	X_\Gamma^\ast\in\widehat{\psi}^{\,\ast}
	\qquad\text{but}\qquad
	X_\Gamma^\ast\notin\widehat{\varphi}^{\,\ast}.
	\]
	Both contradict
	\[
	\widehat{\varphi}^{\,\ast}=\widehat{\psi}^{\,\ast}.
	\]
	Therefore
	\[
	\vdash_{\mathsf{SSL}}\varphi\leftrightarrow\psi.
	\]
\end{proof}

\begin{lemma}[Truth lemma]\label{lem:truth}
	For every formula $\varphi\in\Lang_{KS}$ and every $\Gamma\in\mathsf{MCS}$,
	\[
	M^{\ast c},X_\Gamma^\ast\forces \varphi
	\quad\Longleftrightarrow\quad
	\varphi\in\Gamma.
	\]
\end{lemma}

\begin{proof}
	By induction on the structure of $\varphi$.
	
	\medskip
	\noindent
	\textbf{Case} $\varphi=p$.
	By definition of $\nu^{\ast c}$,
	\[
	M^{\ast c},X_\Gamma^\ast\forces p
	\quad\Longleftrightarrow\quad
	p\in \nu^{\ast c}(X_\Gamma^\ast)
	\quad\Longleftrightarrow\quad
	p\in\Gamma.
	\]
	
	\medskip
	\noindent
	\textbf{Boolean cases.}
	The cases for $\neg$ and $\land$ are immediate by induction and maximal
	consistency.
	
	\medskip
	\noindent
	\textbf{Case} $\varphi=K_a\psi$.
	
	Assume first that $K_a\psi\in\Gamma$. Let $X_\Delta^\ast\in\Fac(M^{\ast c})$
	with $X_\Gamma^\ast\sim_a X_\Delta^\ast$. By
	\cref{lem:aux-canonical-accessibility}, $\Gamma R_a\Delta$. Hence
	$\psi\in\Delta$. By the induction hypothesis,
	\[
	M^{\ast c},X_\Delta^\ast\forces \psi.
	\]
	Since $X_\Delta^\ast$ was arbitrary,
	\[
	M^{\ast c},X_\Gamma^\ast\forces K_a\psi.
	\]
	
	Conversely, assume $K_a\psi\notin\Gamma$. Then $\neg K_a\psi\in\Gamma$. By
	\cref{lem:canonical-existence}, there exists $\Delta\in\mathsf{MCS}$ such
	that
	\[
	\Gamma R_a\Delta
	\qquad\text{and}\qquad
	\neg\psi\in\Delta.
	\]
	By \cref{lem:aux-canonical-accessibility},
	\[
	X_\Gamma^\ast\sim_a X_\Delta^\ast.
	\]
	By the induction hypothesis,
	\[
	M^{\ast c},X_\Delta^\ast\not\forces \psi.
	\]
	Hence
	\[
	M^{\ast c},X_\Gamma^\ast\not\forces K_a\psi.
	\]
	
	Therefore
	\[
	M^{\ast c},X_\Gamma^\ast\forces K_a\psi
	\quad\Longleftrightarrow\quad
	K_a\psi\in\Gamma.
	\]
	
	\medskip
	\noindent
	\textbf{Case} $\varphi=S_a\psi$.
	
	Assume first that $S_a\psi\in\Gamma$. By axiom \textbf{(S1)},
	\[
	K_a\psi\in\Gamma.
	\]
	By the $K_a$-case already proved,
	\[
	M^{\ast c},X_\Gamma^\ast\forces K_a\psi.
	\]
	By the induction hypothesis, the truth set of $\psi$ in $M^{\ast c}$ is
	exactly
	\[
	\llbracket\psi\rrbracket_{M^{\ast c}}=\widehat{\psi}^{\,\ast}.
	\]
	Since $S_a\psi\in\Gamma$, the definition of $N_a^{S,\ast c}$ yields
	\[
	\widehat{\psi}^{\,\ast}\in N_a^{S,\ast c}(v_a(X_\Gamma^\ast)).
	\]
	Hence
	\[
	M^{\ast c},X_\Gamma^\ast\forces S_a\psi.
	\]
	
	Conversely, assume
	\[
	M^{\ast c},X_\Gamma^\ast\forces S_a\psi.
	\]
	Then
	\[
	M^{\ast c},X_\Gamma^\ast\forces K_a\psi
	\qquad\text{and}\qquad
	\llbracket\psi\rrbracket_{M^{\ast c}}
	\in N_a^{S,\ast c}(v_a(X_\Gamma^\ast)).
	\]
	By the $K_a$-case, $K_a\psi\in\Gamma$. By the induction hypothesis,
	\[
	\llbracket\psi\rrbracket_{M^{\ast c}}=\widehat{\psi}^{\,\ast}.
	\]
	Since $\widehat{\psi}^{\,\ast}\in N_a^{S,\ast c}(v_a(X_\Gamma^\ast))$, the
	definition of $N_a^{S,\ast c}$ implies that there exists some formula
	$\theta$ such that
	\[
	S_a\theta\in\Gamma
	\qquad\text{and}\qquad
	\widehat{\theta}^{\,\ast}=\widehat{\psi}^{\,\ast}.
	\]
	By \cref{lem:extensionality},
	\[
	\vdash_{\mathsf{SSL}}\theta\leftrightarrow\psi.
	\]
	By the rule of Replacement of Equivalents for Secrecy,
	\[
	\vdash_{\mathsf{SSL}}S_a\theta\leftrightarrow S_a\psi.
	\]
	Since $S_a\theta\in\Gamma$ and $\Gamma$ is maximally consistent, it follows
	that
	\[
	S_a\psi\in\Gamma.
	\]
	
	This completes the induction.
\end{proof}

\subsection{Share representation theorem}

We now eliminate the auxiliary colour and return to pure $A$-chromatic models.
The representation step is in the spirit of geometric translation techniques
used elsewhere in simplicial epistemic logic, where semantic information is
redistributed over chromatic structure rather than kept in an external label or
world component~\cite{goubault2021simplicial,goubault2022kb4n,
	randrianomentsoa2023impure}.

\begin{definition}[Share model]\label{def:share-model}
	Let
	\[
	M^\ast=(V,\mathcal{F},\chi,\nu,\{N_a^S\}_{a\in A})
	\]
	be an auxiliary-colour simplicial secrecy model over $A$.
	
	Let
	\[
	G:=\mathbb{Z}^{(\Fac(M^\ast))}
	\]
	be the free abelian group on $\Fac(M^\ast)$, and let
	\[
	\iota:\Fac(M^\ast)\to G
	\]
	be the canonical injection sending each facet $X$ to the corresponding basis
	vector $e_X$.
	
	For each facet $X\in\Fac(M^\ast)$ and each function $\sigma:A\to G$
	satisfying
	\[
	\sum_{a\in A}\sigma(a)=\iota(X),
	\]
	define the $A$-coloured facet
	\[
	X^\sigma:=\{(a,v_a(X),\sigma(a))\mid a\in A\}.
	\]
	
	The \emph{share model} associated with $M^\ast$ is the structure
	\[
	\mathsf{Sh}(M^\ast)=
	(V^{\mathrm{sh}},\mathcal{F}^{\mathrm{sh}},\chi^{\mathrm{sh}},
	\nu^{\mathrm{sh}},\{N_a^{S,\mathrm{sh}}\}_{a\in A})
	\]
	defined as follows.
	\begin{enumerate}
		\item The set of facets is
		\[
		\Fac(\mathsf{Sh}(M^\ast))
		=
		\{X^\sigma \mid X\in\Fac(M^\ast),\ \sigma:A\to G,\
		\sum_{a\in A}\sigma(a)=\iota(X)\}.
		\]
		The full face set $\mathcal{F}^{\mathrm{sh}}$ is the downward closure of
		$\Fac(\mathsf{Sh}(M^\ast))$.
		
		\item The set of vertices is
		\[
		V^{\mathrm{sh}}
		=
		\{(a,v_a(X),\sigma(a)) \mid X^\sigma\in\Fac(\mathsf{Sh}(M^\ast)),\ a\in A\}.
		\]
		
		\item The colouring function is given by
		\[
		\chi^{\mathrm{sh}}((a,v,g))=a.
		\]
		
		\item The valuation on facets is given by
		\[
		\nu^{\mathrm{sh}}(X^\sigma)=\nu(X).
		\]
		
		\item For each $U\subseteq\Fac(M^\ast)$, define its \emph{lift}
		\[
		U^\uparrow
		:=
		\{Y^\tau\in\Fac(\mathsf{Sh}(M^\ast))\mid Y\in U\}.
		\]
		
		\item For each vertex $(a,v,g)\in V^{\mathrm{sh}}$, define
		\[
		N_a^{S,\mathrm{sh}}((a,v,g))
		:=
		\{U^\uparrow \mid U\in N_a^S(v)\}.
		\]
	\end{enumerate}
\end{definition}

\begin{lemma}[Uniqueness of share representation]\label{lem:share-unique}
	Let $X,Y\in\Fac(M^\ast)$, and let $\sigma,\tau:A\to G$ satisfy
	\[
	\sum_{a\in A}\sigma(a)=\iota(X),
	\qquad
	\sum_{a\in A}\tau(a)=\iota(Y).
	\]
	If
	\[
	X^\sigma=Y^\tau,
	\]
	then
	\[
	X=Y
	\qquad\text{and}\qquad
	\sigma=\tau.
	\]
\end{lemma}

\begin{proof}
	Since both $X^\sigma$ and $Y^\tau$ are facets containing exactly one vertex
	of each colour $a\in A$, for every $a\in A$ we have
	\[
	(a,v_a(X),\sigma(a))=(a,v_a(Y),\tau(a)).
	\]
	Hence
	\[
	v_a(X)=v_a(Y)
	\qquad\text{and}\qquad
	\sigma(a)=\tau(a)
	\]
	for all $a\in A$. Therefore $\sigma=\tau$. It follows that
	\[
	\iota(X)=\sum_{a\in A}\sigma(a)=\sum_{a\in A}\tau(a)=\iota(Y).
	\]
	Since $\iota$ is injective, we obtain $X=Y$.
\end{proof}

\begin{lemma}[Injectivity of lifting]\label{lem:lift-injective}
	For all $U,W\subseteq\Fac(M^\ast)$,
	\[
	U^\uparrow=W^\uparrow
	\quad\Longrightarrow\quad
	U=W.
	\]
\end{lemma}

\begin{proof}
	Suppose $U^\uparrow=W^\uparrow$. Let $X\in\Fac(M^\ast)$ be arbitrary. Since
	$A\neq\emptyset$, choose some $a_0\in A$ and define $\sigma_X:A\to G$ by
	\[
	\sigma_X(a_0)=\iota(X),
	\qquad
	\sigma_X(c)=0 \text{ for all } c\in A\setminus\{a_0\}.
	\]
	Then
	\[
	\sum_{a\in A}\sigma_X(a)=\iota(X),
	\]
	so
	\[
	X^{\sigma_X}\in\Fac(\mathsf{Sh}(M^\ast)).
	\]
	
	We claim that for every $Z\subseteq\Fac(M^\ast)$,
	\[
	X\in Z
	\quad\Longleftrightarrow\quad
	X^{\sigma_X}\in Z^\uparrow.
	\]
	The left-to-right direction is immediate from the definition of $Z^\uparrow$.
	For the right-to-left direction, if
	\[
	X^{\sigma_X}\in Z^\uparrow,
	\]
	then there exist $Y\in Z$ and $\tau:A\to G$ such that
	\[
	Y^\tau\in\Fac(\mathsf{Sh}(M^\ast))
	\qquad\text{and}\qquad
	X^{\sigma_X}=Y^\tau.
	\]
	By \cref{lem:share-unique}, it follows that $X=Y$, hence $X\in Z$.
	
	Applying this claim with $Z=U$ and $Z=W$, we obtain
	\[
	X\in U
	\quad\Longleftrightarrow\quad
	X^{\sigma_X}\in U^\uparrow
	\quad\Longleftrightarrow\quad
	X^{\sigma_X}\in W^\uparrow
	\quad\Longleftrightarrow\quad
	X\in W.
	\]
	Since $X$ was arbitrary, $U=W$.
\end{proof}

\begin{lemma}[Share-completion lemma]\label{lem:share-completion}
	Let $X^\sigma\in\Fac(\mathsf{Sh}(M^\ast))$, let $a\in A$, and let
	$Y\in\Fac(M^\ast)$ be such that
	\[
	X\sim_a Y
	\]
	in the auxiliary model $M^\ast$. Then there exists a function $\tau:A\to G$
	such that
	\[
	Y^\tau\in\Fac(\mathsf{Sh}(M^\ast))
	\qquad\text{and}\qquad
	\tau(a)=\sigma(a).
	\]
\end{lemma}

\begin{proof}
	Since $|A|\ge 2$, choose some $b\in A\setminus\{a\}$. Define
	$\tau:A\to G$ by
	\[
	\tau(a)=\sigma(a),\qquad
	\tau(b)=\iota(Y)-\sigma(a),\qquad
	\tau(c)=0 \text{ for } c\in A\setminus\{a,b\}.
	\]
	Then
	\[
	\sum_{c\in A}\tau(c)=\iota(Y),
	\]
	so $Y^\tau\in\Fac(\mathsf{Sh}(M^\ast))$, and by construction
	$\tau(a)=\sigma(a)$.
\end{proof}

\begin{lemma}\label{lem:share-model-ssm}
	If $M^\ast$ is an auxiliary-colour simplicial secrecy model over $A$, then
	$\mathsf{Sh}(M^\ast)$ is a simplicial secrecy model in the original sense of
	\cref{def:ssm}.
\end{lemma}

\begin{proof}
	We first verify that the underlying simplicial part is a pure $A$-chromatic
	simplicial epistemic model.
	
	Every facet $X^\sigma$ contains exactly one vertex of each colour $a\in A$,
	so it is $A$-chromatic and complete. Moreover, by the definition of
	$V^{\mathrm{sh}}$, every vertex of $V^{\mathrm{sh}}$ belongs to some facet
	of $\mathsf{Sh}(M^\ast)$. It remains to check that the valuation is well
	defined. Suppose
	\[
	X^\sigma=Y^\tau.
	\]
	By \cref{lem:share-unique}, we obtain $X=Y$. Hence
	\[
	\nu^{\mathrm{sh}}(X^\sigma)=\nu(X)=\nu(Y)=\nu^{\mathrm{sh}}(Y^\tau).
	\]
	
	It remains to verify \textbf{(SN)}.
	
	Let $(a,v,g)\in V^{\mathrm{sh}}$ and let
	\[
	U^\uparrow\in N_a^{S,\mathrm{sh}}((a,v,g)),
	\]
	where $U\in N_a^S(v)$.
	
	Let $X^\sigma\in\St((a,v,g))$ and let $b\in A\setminus\{a\}$. Then
	\[
	(a,v,g)\in X^\sigma,
	\]
	so
	\[
	v_a(X)=v
	\qquad\text{and}\qquad
	\sigma(a)=g.
	\]
	Hence $X\in\St(v)$ in the auxiliary model $M^\ast$.
	
	Since $U\in N_a^S(v)$ and $M^\ast$ satisfies \textbf{(SN)}, there exists
	$Y\in\Fac(M^\ast)$ such that
	\[
	X\sim_b Y
	\qquad\text{and}\qquad
	Y\notin U.
	\]
	By \cref{lem:share-completion}, there exists $\tau:A\to G$ such that
	\[
	Y^\tau\in\Fac(\mathsf{Sh}(M^\ast))
	\qquad\text{and}\qquad
	\tau(b)=\sigma(b).
	\]
	It follows that
	\[
	X^\sigma\sim_b Y^\tau
	\]
	in the share model, because the $b$-vertices of $X^\sigma$ and $Y^\tau$ are
	both
	\[
	(b,v_b(X),\sigma(b))=(b,v_b(Y),\tau(b)).
	\]
	Since $Y\notin U$, we have $Y^\tau\notin U^\uparrow$. Therefore
	\textbf{(SN)} holds.
\end{proof}

\begin{theorem}[Share representation theorem]\label{thm:representation}
	Let $M^\ast$ be an auxiliary-colour simplicial secrecy model over $A$. For
	every formula $\varphi\in\Lang_{KS}$, every facet $X\in\Fac(M^\ast)$, and
	every function $\sigma:A\to G$ satisfying
	\[
	\sum_{a\in A}\sigma(a)=\iota(X),
	\]
	we have
	\[
	\mathsf{Sh}(M^\ast),X^\sigma\forces \varphi
	\quad\Longleftrightarrow\quad
	M^\ast,X\forces \varphi.
	\]
\end{theorem}

\begin{proof}
	By induction on the structure of $\varphi$.
	
	\medskip
	\noindent
	\textbf{Case} $\varphi=p$.
	By definition of the valuation on the share model,
	\[
	\mathsf{Sh}(M^\ast),X^\sigma\forces p
	\quad\Longleftrightarrow\quad
	p\in \nu^{\mathrm{sh}}(X^\sigma)
	\quad\Longleftrightarrow\quad
	p\in \nu(X)
	\quad\Longleftrightarrow\quad
	M^\ast,X\forces p.
	\]
	
	\medskip
	\noindent
	\textbf{Boolean cases.}
	The cases for $\neg$ and $\land$ are immediate by the induction hypothesis.
	
	\medskip
	\noindent
	\textbf{Case} $\varphi=K_a\psi$.
	
	Assume first that $M^\ast,X\forces K_a\psi$. Let
	$Y^\tau\in\Fac(\mathsf{Sh}(M^\ast))$ be such that
	\[
	X^\sigma\sim_a Y^\tau.
	\]
	Then the $a$-vertices of $X^\sigma$ and $Y^\tau$ are equal, so in particular
	\[
	v_a(X)=v_a(Y).
	\]
	Hence $X\sim_a Y$ in $M^\ast$. Since $M^\ast,X\forces K_a\psi$, we get
	\[
	M^\ast,Y\forces \psi.
	\]
	By the induction hypothesis,
	\[
	\mathsf{Sh}(M^\ast),Y^\tau\forces \psi.
	\]
	Since $Y^\tau$ was arbitrary,
	\[
	\mathsf{Sh}(M^\ast),X^\sigma\forces K_a\psi.
	\]
	
	Conversely, assume
	\[
	M^\ast,X\not\forces K_a\psi.
	\]
	Then there exists $Y\in\Fac(M^\ast)$ such that
	\[
	X\sim_a Y
	\qquad\text{and}\qquad
	M^\ast,Y\not\forces \psi.
	\]
	By \cref{lem:share-completion}, there exists $\tau:A\to G$ such that
	\[
	Y^\tau\in\Fac(\mathsf{Sh}(M^\ast))
	\qquad\text{and}\qquad
	\tau(a)=\sigma(a).
	\]
	Therefore
	\[
	X^\sigma\sim_a Y^\tau
	\]
	in the share model. By the induction hypothesis,
	\[
	\mathsf{Sh}(M^\ast),Y^\tau\not\forces \psi.
	\]
	Hence
	\[
	\mathsf{Sh}(M^\ast),X^\sigma\not\forces K_a\psi.
	\]
	
	Thus
	\[
	\mathsf{Sh}(M^\ast),X^\sigma\forces K_a\psi
	\quad\Longleftrightarrow\quad
	M^\ast,X\forces K_a\psi.
	\]
	
	\medskip
	\noindent
	\textbf{Case} $\varphi=S_a\psi$.
	
	By the $K_a$-case already proved,
	\[
	\mathsf{Sh}(M^\ast),X^\sigma\forces K_a\psi
	\quad\Longleftrightarrow\quad
	M^\ast,X\forces K_a\psi.
	\]
	
	By the induction hypothesis, for every facet
	$Y^\tau\in\Fac(\mathsf{Sh}(M^\ast))$,
	\[
	\mathsf{Sh}(M^\ast),Y^\tau\forces \psi
	\quad\Longleftrightarrow\quad
	M^\ast,Y\forces \psi.
	\]
	Hence the truth set of $\psi$ in the share model is exactly the lift of its
	truth set in the auxiliary model:
	\[
	\llbracket\psi\rrbracket_{\mathsf{Sh}(M^\ast)}
	=
	\bigl(\llbracket\psi\rrbracket_{M^\ast}\bigr)^\uparrow.
	\]
	
	Now $v_a(X^\sigma)=(a,v_a(X),\sigma(a))$, and by definition of the share
	neighborhoods,
	\[
	N_a^{S,\mathrm{sh}}(v_a(X^\sigma))
	=
	\{U^\uparrow \mid U\in N_a^S(v_a(X))\}.
	\]
	Therefore
	\begin{align*}
		\llbracket\psi\rrbracket_{\mathsf{Sh}(M^\ast)}
		\in N_a^{S,\mathrm{sh}}(v_a(X^\sigma))
		&\Longleftrightarrow
		\bigl(\llbracket\psi\rrbracket_{M^\ast}\bigr)^\uparrow
		\in \{U^\uparrow \mid U\in N_a^S(v_a(X))\}\\
		&\Longleftrightarrow
		\llbracket\psi\rrbracket_{M^\ast}\in N_a^S(v_a(X)),
	\end{align*}
	where the second equivalence uses \cref{lem:lift-injective}.
	
	Combining the knowledge part with the neighborhood part, we obtain
	\[
	\mathsf{Sh}(M^\ast),X^\sigma\forces S_a\psi
	\quad\Longleftrightarrow\quad
	M^\ast,X\forces S_a\psi.
	\]
	
	This completes the induction.
\end{proof}

\subsection{Completeness}

We can now prove completeness by contraposition.

\begin{theorem}[Completeness of $\mathsf{SSL}$ for \texorpdfstring{$|A|\ge 2$}{|A|>=2}]\label{thm:completeness}
	For every formula $\varphi\in\Lang_{KS}$,
	\[
	\models \varphi
	\quad\Longrightarrow\quad
	\vdash_{\mathsf{SSL}}\varphi.
	\]
\end{theorem}

\begin{proof}
	Suppose
	\[
	\not\vdash_{\mathsf{SSL}}\varphi.
	\]
	Then $\neg\varphi$ is $\mathsf{SSL}$-consistent, and can be extended to some
	$\Gamma\in\mathsf{MCS}$.
	
	By \cref{lem:truth},
	\[
	M^{\ast c},X_\Gamma^\ast\forces \neg\varphi.
	\]
	Hence
	\[
	M^{\ast c},X_\Gamma^\ast\not\forces \varphi.
	\]
	
	Let
	\[
	G^c:=\mathbb{Z}^{(\Fac(M^{\ast c}))}
	\]
	and let
	\[
	\iota^c:\Fac(M^{\ast c})\to G^c
	\]
	be the canonical injection used in the share construction
	\[
	\mathsf{Sh}(M^{\ast c}).
	\]
	Choose any function $\sigma_\Gamma:A\to G^c$ such that
	\[
	\sum_{a\in A}\sigma_\Gamma(a)=\iota^c(X_\Gamma^\ast);
	\]
	for example, pick some $a_0\in A$ and let
	\[
	\sigma_\Gamma(a_0)=\iota^c(X_\Gamma^\ast),
	\qquad
	\sigma_\Gamma(c)=0 \text{ for } c\neq a_0.
	\]
	Then
	\[
	(X_\Gamma^\ast)^{\sigma_\Gamma}\in\Fac(\mathsf{Sh}(M^{\ast c})).
	\]
	By \cref{thm:representation},
	\[
	\mathsf{Sh}(M^{\ast c}),(X_\Gamma^\ast)^{\sigma_\Gamma}\not\forces \varphi.
	\]
	By \cref{lem:share-model-ssm}, $\mathsf{Sh}(M^{\ast c})$ is a simplicial
	secrecy model in the original sense of \cref{def:ssm}. Therefore $\varphi$ is
	not valid on the original class of simplicial secrecy models.
	
	By contraposition,
	\[
	\models \varphi \Longrightarrow \vdash_{\mathsf{SSL}}\varphi.
	\]
\end{proof}

\begin{remark}
	The completeness proof has two distinct components. The auxiliary canonical
	model $M^{\ast c}$ supplies the usual maximally-consistent-set machinery,
	familiar from canonical constructions in modal and neighborhood
	logics~\cite{chellas1980modal,vanbenthem2014evidence}, while avoiding the
	collapse of distinct worlds into the same $A$-coloured facet. The share
	representation theorem then removes the auxiliary colour by redistributing its
	global information over the original agent colours, in a manner consonant
	with geometric representation techniques used elsewhere in simplicial
	epistemic logic~\cite{goubault2021simplicial,goubault2022kb4n,
		randrianomentsoa2023impure}. In this way, the final completeness theorem is
	obtained for the original class of pure $A$-chromatic simplicial secrecy
	models.
\end{remark}

\subsection{Strong completeness and the role of \texorpdfstring{$|A|\ge 2$}{|A|>=2}}

For semantic consequence, we write
\[
\Phi \models \varphi
\]
to mean that for every simplicial secrecy model $M$ and every facet $X\in\Fac(M)$,
if
\[
M,X\forces \psi
\qquad\text{for all } \psi\in\Phi,
\]
then
\[
M,X\forces \varphi.
\]

\begin{definition}[Derivability from assumptions]\label{def:derivability-from-assumptions}
	For $\Phi\subseteq\Lang_{KS}$ and $\varphi\in\Lang_{KS}$, we write
	\[
	\Phi\vdash_{\mathsf{SSL}}\varphi
	\]
	if and only if there exist formulas $\phi_1,\dots,\phi_n\in\Phi$ such that
	\[
	\vdash_{\mathsf{SSL}}(\phi_1\land\dots\land\phi_n)\rightarrow\varphi,
	\]
	where $n\ge 0$, and the empty conjunction is understood as $\top$.
\end{definition}

\begin{corollary}[Strong completeness of $\mathsf{SSL}$ for \texorpdfstring{$|A|\ge 2$}{|A|>=2}]\label{cor:strong-completeness}
	For every set of formulas $\Phi\subseteq\Lang_{KS}$ and every formula
	$\varphi\in\Lang_{KS}$,
	\[
	\Phi \models \varphi
	\quad\Longrightarrow\quad
	\Phi \vdash_{\mathsf{SSL}} \varphi.
	\]
\end{corollary}

\begin{proof}
	Suppose
	\[
	\Phi \not\vdash_{\mathsf{SSL}}\varphi.
	\]
	We claim that
	\[
	\Phi\cup\{\neg\varphi\}
	\]
	is $\mathsf{SSL}$-consistent in the finitary sense.
	
	For otherwise, there would exist formulas $\phi_1,\dots,\phi_n\in\Phi$ such
	that
	\[
	\vdash_{\mathsf{SSL}}
	(\phi_1\land\dots\land\phi_n\land\neg\varphi)\rightarrow\bot.
	\]
	By propositional reasoning, this yields
	\[
	\vdash_{\mathsf{SSL}}
	(\phi_1\land\dots\land\phi_n)\rightarrow\varphi,
	\]
	contradicting
	\[
	\Phi\not\vdash_{\mathsf{SSL}}\varphi.
	\]
	Hence $\Phi\cup\{\neg\varphi\}$ is consistent.
	
	Extend it to some $\Gamma\in\mathsf{MCS}$. Then
	\[
	\Phi\subseteq\Gamma
	\qquad\text{and}\qquad
	\neg\varphi\in\Gamma.
	\]
	By \cref{lem:truth},
	\[
	M^{\ast c},X_\Gamma^\ast\forces \Phi
	\qquad\text{and}\qquad
	M^{\ast c},X_\Gamma^\ast\not\forces \varphi.
	\]
	
	Let
	\[
	G^c:=\mathbb{Z}^{(\Fac(M^{\ast c}))}
	\]
	and let
	\[
	\iota^c:\Fac(M^{\ast c})\to G^c
	\]
	be the canonical injection used in the share construction
	\[
	\mathsf{Sh}(M^{\ast c}).
	\]
	Choose any $\sigma_\Gamma:A\to G^c$ such that
	\[
	\sum_{a\in A}\sigma_\Gamma(a)=\iota^c(X_\Gamma^\ast).
	\]
	By \cref{thm:representation},
	\[
	\mathsf{Sh}(M^{\ast c}),(X_\Gamma^\ast)^{\sigma_\Gamma}\forces \Phi
	\qquad\text{and}\qquad
	\mathsf{Sh}(M^{\ast c}),(X_\Gamma^\ast)^{\sigma_\Gamma}\not\forces \varphi.
	\]
	By \cref{lem:share-model-ssm}, this is a countermodel in the original class
	of simplicial secrecy models. Hence
	\[
	\Phi\not\models \varphi.
	\]
	Therefore, by contraposition,
	\[
	\Phi\models\varphi
	\quad\Longrightarrow\quad
	\Phi\vdash_{\mathsf{SSL}}\varphi.
	\]
\end{proof}

\begin{remark}\label{rem:one-agent-role}
	The restriction $|A|\ge 2$ is essential for the completeness results above.
	If $A=\{a\}$, then every pure $A$-chromatic simplicial epistemic model has
	singleton facets, so $\sim_a$ is the identity relation and
	\[
	K_a\varphi\leftrightarrow\varphi
	\]
	is valid on the whole model class. Thus the one-agent class is too small to
	serve as a complete semantics for the full S5 part of $\mathsf{SSL}$.
	
	Technically, the same restriction is reflected in the share construction:
	the share-completion lemma requires a second agent colour in order to
	redistribute the auxiliary global information while keeping one designated
	share fixed. For this reason, the metatheoretic results proved in this
	section are naturally stated for the genuinely multi-agent case.
\end{remark}

\section{Conclusion and Future Work}\label{sec:conclusion}

\subsection{Summary of the paper}

In this paper we developed a logic of secrecy on simplicial models. The central
idea has been that secrecy should not be treated merely as a definable
combination of knowledge and ignorance. Instead, we introduced \emph{simplicial
	secrecy models}, which preserve the standard simplicial semantics for epistemic
knowledge while enriching it with a vertex-based neighborhood layer attached to
agents' local states. This yields a primitive secrecy operator $S_a$ whose truth
depends both on ordinary simplicial knowledge and on whether the truth set of a
formula is designated as secret at the owner's current local state. On top of
this semantics, we formulated the axiomatic system $\mathsf{SSL}$, showing that
it captures the owner-local yet genuinely non-normal behavior of secrecy, and we
established soundness for the full class of simplicial secrecy models together
with completeness, for the genuinely multi-agent case $|A|\ge 2$, by means of
an auxiliary-colour canonical construction and a representation theorem back
into pure $A$-chromatic simplicial secrecy models. In this way, the paper
provides a primitive, vertex-based, and geometrically grounded account of
secrecy in simplicial epistemic semantics.

\subsection{Future directions}

A first direction concerns the proof theory of secrecy over restricted classes
of secrecy neighborhoods. The present system $\mathsf{SSL}$ is exact for the
current semantics in the genuinely multi-agent case $|A|\ge 2$, but it
remains to determine which additional principles become valid once one
imposes natural structural constraints on the secrecy neighborhoods. This
would extend the present analysis in the broader spirit of refined
neighborhood proof theory~\cite{chellas1980modal,vanbenthem2014evidence}.

A second direction is \emph{coalition secrecy}. In many multi-agent settings, a
proposition may remain secret from each individual agent while ceasing to be
secret from a group that pools its local information. Since simplicial models
are especially well suited to representing combinations of local states, they
provide a natural setting for studying secrecy relative to coalitions and its
interaction with distributed knowledge~\cite{fagin1995reasoning,
	goubault2023semisimplicial,galimullin2024varieties,
	balbiani2024dynamicdistributed}.

A third direction is \emph{secrecy under failures}. One of the main motivations
for simplicial semantics comes from distributed computation, where crash
failures, omissions, and partial communication failures are central
phenomena~\cite{herlihy2014distributed}. Extending the present framework to
fault-sensitive simplicial settings may reveal new connections between secrecy,
loss of local information, and epistemic robustness under failures; existing
simplicial work on epistemic logic with failures suggests a natural starting
point for such an investigation~\cite{goubault2022kb4n}.

A fourth direction is a \emph{dynamic logic of secrecy on simplicial models}.
Secrets in multi-agent systems are rarely static: they may be revealed,
preserved, weakened, or transformed by public announcements, private
communication, or protocol executions. It would therefore be valuable to study
how the vertex-based secrecy neighborhoods introduced here evolve under
epistemic updates, building on existing dynamic and communication-oriented
developments in simplicial and distributed epistemic
logic~\cite{goubault2021simplicial,vanditmarsch2021equality,
	velazquez2021communication,castaneda2024pattern}.

More generally, we hope that the framework developed here helps clarify how
secrecy should be understood in local-state semantics. In Kripke-style
approaches, secrecy is typically analyzed as a pattern of knowledge and
ignorance over worlds. In the simplicial setting proposed here, it becomes a
structural feature of an agent's position in a combinatorial epistemic space.
This suggests that secrecy is not merely a derivative epistemic condition, but
also a geometric phenomenon of distributed information.


\end{document}